\documentclass[12pt]{article}
\usepackage{amssymb}

\usepackage{natbib,graphicx,setspace,lscape,longtable}
\usepackage{natbib,epsfig,graphicx}
\usepackage{amsmath,amsthm,amssymb,color, bbm}
\usepackage{enumerate}
\usepackage{mathtools}
\usepackage{array}
\usepackage{booktabs}
\usepackage{rotating}

\usepackage{float}
\usepackage{hyperref}
\hypersetup{
linkcolor  = blue,
citecolor  =blue,
urlcolor   = blue,
colorlinks = true,
}
\usepackage[ruled,vlined]{algorithm2e}
\bibpunct{(}{)}{;}{a}{,}{,}

\newtheorem{theorem}{Theorem}
\newtheorem{lemma}{Lemma}

\newcommand{\csection}[1]
{\begin{center}
	\stepcounter{section}
	{\bf\large\arabic{section}. #1}
\end{center}
\vspace{-0.15 cm}
}

\newcommand{\scsection}[1]
{\begin{center}
	{\bf\large #1}
\end{center}
\vspace{-0.15 cm}
}

\newcommand{\csubsection}[1]{
\vspace{-0.25 cm}
\begin{center}
	\stepcounter{subsection}
	{\it\arabic{section}.\arabic{subsection}. #1}
\end{center}
\vspace{-0.25 cm}
}

\newcommand{\scsubsection}[1]{
\vspace{-0.25 cm}
\begin{center}
	\stepcounter{subsection}
	{\it #1}
\end{center}
\vspace{-0.25 cm}
}

\def\beq{\begin{equation}}
\def\eeq{\end{equation}}
\def\beqr{\begin{eqnarray}}
\def\eeqr{\end{eqnarray}}
\def\beqrs{\begin{eqnarray*}}
\def\eeqrs{\end{eqnarray*}}
\def\bet{\begin{theorem}}
\def\eet{\end{theorem}}
\def\bel{\begin{lemma}}
\def\eel{\end{lemma}}
\def\bg{\begin{figure}[tbph]\begin{center}}
	\def\eg{\end{center}\end{figure}}

\def\bc{\begin{center}}
\def\ec{\end{center}}
\usepackage{amsmath}
\usepackage{accents}
\usepackage{multirow}
\makeatletter
\def\widebar{\accentset{{\cc@style\underline{\mskip10mu}}}}
\def\Widebar{\accentset{{\cc@style\underline{\mskip8mu}}}}
\makeatother

\def\1{\mbox{\boldmath $1$}}
\def\0{\boldsymbol{0}}

\def\mL{\mathcal L}

\def\mR{\mathbb{R}}
\def\mS{\mathcal S}

\def\rank{\mbox{rank}}

\def\tr{\mbox{tr}}
\def\diag{\mbox{diag}}

\def\|{\Vert}
\def\wt{\widetilde}

\def\boxit#1{\vbox{\hrule\hbox{\vrule\kern6pt\vbox{\kern6pt#1\kern6pt}\kern6pt\vrule}\hrule}}

\def\wh{\widehat}

\numberwithin{equation}{section}

{\newcommand{\killproofname}{\unskip\nopunct}}
{\newcommand{\killproofname}[1]{\unskip\aftergroup\ignorespaces\ignorespaces}}
\makeatother

\begin{document}

\begin{center}
	{\bf\Large An Asymptotic Analysis of Minibatch-Based Momentum Methods for Linear Regression Models}\\
	Yuan Gao$ ^{a,b} $, Xuening Zhu$ ^c $, Haobo Qi$ ^d, $ Guodong Li$^e  $, Riquan Zhang$ ^{a,b} $, Hansheng Wang$ ^d $\\
	{\it\small
		$ ^a $ School of Statistics, East China Normal University, Shanghai, China\\
		$ ^b $ Key Laboratory of Advanced Theory and Application in Statistics and Data Science - MOE, East China Normal University, Shanghai, China\\
		$ ^c $School of Data Science, Fudan University, Shanghai, China\\
		$ ^d $Guanghua School of Management, Peking University, Beijing, China\\
		$ ^e $Department of Statistics and Actuarial Science, University of Hong Kong, Hong Kong, China}

	
\end{center}

\begin{abstract}
	Momentum methods have been shown to accelerate the convergence of the standard gradient descent algorithm in practice and theory.
	In particular, the minibatch-based gradient descent methods with momentum (MGDM) are widely used to solve large-scale optimization problems with massive datasets.
	Despite the success of the MGDM methods in practice, their theoretical properties are still underexplored. To this end, we investigate the theoretical properties of MGDM methods based on the linear regression models. We first study the numerical convergence properties of the MGDM algorithm and further provide the theoretically optimal tuning parameters specification to achieve faster convergence rate. In addition, we explore the relationship between the statistical properties of the resulting MGDM estimator and the tuning parameters. Based on these theoretical findings, we give the conditions for the resulting estimator to achieve the optimal statistical efficiency.
	Finally, extensive numerical experiments are conducted to verify our theoretical results.

\end{abstract}
\textbf{KEYWORDS:} Gradient Descent, Momentum Method, Fixed Minibatch, Shuffled Minibatch

\newpage

\csection{INTRODUCTION}

Modern statistical analysis often involves massive datasets and sophisticated models, which makes the corresponding statistical computation extremely challenging. For example, the well-known \textit{ImageNet} dataset, {contains} more than one million photos and takes up approximately 150GB space on a hard drive \citep{deng2009imagenet}. Obviously, most computers available to researchers cannot easily load such an enormous dataset as a whole into its memory. Then, how to process these massive datasets efficiently with statistical guarantees becomes a problem of great interest. To address this issue, it is in urgent need to develop novel scalable algorithms for statistical analysis.

In recent yeas, a great number of researches have been devoted to dealing with large-scale statistical optimization problems. For example, the distributed computing methods, based on the divide-and-conquer strategy from computer science literature, become very popular in the statistics community.
Many statistical learning problems have been investigated on a distributed system, such as generalized linear models \citep{chen2014split, battey2018distributed, tang2020distributed}, quantile regression \citep{chen2019quantile, chen2020distributed}, kernel ridge regression \citep{chang2017distributed,lin2020distributed}, support vector machine \citep{wang2019distributed}, principal component analysis \citep{fan2019distributed, chen2021distributed}, and references therein.
Clearly, the distributed computing frameworks reduce the computational burden of  one single machine by employing a computer cluster.
However, it still replies on a distributed architecture for implementation, which can be expensive for many practical problems.
With limited computational resources, it is desirable 
to conduct computation with a single computer and feasible optimization algorithms.
In this regard, stochastic optimization methods, especially stochastic gradient descent (SGD) and its derivatives, play an important role.

Before introducing SGD and its derivatives, we first discuss their classical versions.
As a first-order optimization method, gradient descent (GD) is arguably the simplest one.
Unlike second-order methods, such as Newton-Raphson algorithm, a first-order method requires no computation of the Hessian matrices.
In fact, a standard GD algorithm only requires the gradients of the objective function and a learning rate. 
Therefore, it is computationally more feasible for sophisticated models with massive datasets. Furthermore, the cost in computer memory for storing Hessian matrices can be significantly reduced.
Unfortunately, there is no free lunch. The tradeoff for GD methods' excellent practical feasibility is slow convergence. In other words, the convergence rate of a GD algorithm is usually much slower than that of a Newton-Raphson algorithm, if the Newton-Raphson algorithm can be implemented. This is particularly true for ill-conditioned problems \citep{polyak1987introduction}. In this case, feature directions with tiny variability are likely to be ignored by a standard GD algorithm. To partially fix this problem, \cite{polyak1964} proposed a heavy-ball method to accelerate the standard GD algorithm. Specifically, it uses the information from not only the current step but also the previous step (i.e., the momentum). As a consequence, the search direction can be slightly modified. With this slight but novel modification, the convergence rate of gradient descent with momentum (GDM) could be much improved compared to the standard GD algorithm.

Despite its theoretical popularity, the practical implementation of GDM could be nontrivial on a massive dataset. Under a traditional setup with small sample size and simple model structure, a GDM algorithm can be easily implemented by loading the entire dataset into the computer memory as a whole. Subsequently, various gradients and momentums can be efficiently computed and processed. Unfortunately, such a common practice could be seriously challenged if the target model structure is sufficiently complex and the dataset is massive in size. For example, the \textit{ResNet50} model \citep{he2016deep} is a classical deep learning model with more than 25 million parameters. In this case, it would take a long time to compute the gradient and momentum needed for performing one single step of iteration. Thus, implementing GDM algorithms on such a complex model (i.e., the \textit{ResNet50} model) with such a massive dataset (i.e., the \textit{ImageNet} dataset) is very challenging.

Since a massive dataset cannot be easily loaded into computer memory as a whole, it has to be processed in a batch-by-batch manner. Here, ``batch'' refers to a subsample of the original whole dataset. Subsequently, the important quantities (e.g., gradient and momentum) can be computed based on the batch. This approach leads to various stochastic gradient descent (SGD) methods. Here, ``stochastic" refers to the batch being randomly subsampled from the original whole dataset.
Then, according to how this subsample (i.e., batch) is generated and utilized, different SGD algorithms can be developed. The history of SGD can date back to the pioneering work of \cite{robbins1951stochastic}. 
The computational efficiency of SGD finds wide application in the large-scale optimization problems, and attracts a large number of follow-up studies \citep[see, e.g.,][and references therein]{wang2013variance, wang2017stochastic, yang2019multilevel, chen2021first}. 
Although many existing works focus on the numerical convergence rate of the stochastic optimization algorithms, there are also a few of them that investigate the resulting estimator from a statistical viewpoint.
For example, \cite{toulis2017asymptotic} proposed an implicit SGD procedure, and analyzed the asymptotic behavior of the estimators resulting from the standard and implicit SGD procedures. \cite{chen2020statistical} and \cite{zhu2021online} investigated the statistical inference problem based on the averaged SGD estimator, and provided several methods to estimate the asymptotic covariance.
\cite{luo2020renewable} developed a renewable estimation procedure that extends the standard SGD and averaged implicit SGD algorithm, and further studied the asymptotic properties of the resulting estimator.
In general, the algorithms that they discussed are one-pass in nature and should be very useful for streaming data analysis. However, for non-streaming situations where the massive dataset is already on the hard drive, the one-pass type of algorithm is less frequently used in practice. Instead, multiple passes (or epochs) are more commonly used for better statistical efficiency.

In fact, the many popularly used SGD algorithms in practice are based on shuffled minibatches. In other words, the entire SGD algorithm is divided into different epochs. For each epoch, the original whole dataset is randomly shuffled and then sequentially divided into different batches. Subsequently, different batches are sequentially loaded into the computer memory and then processed. Once all the batches have been successfully executed, the whole dataset might need to be reshuffled. This process accomplishes one epoch of the minibatch-based SGD algorithm. Accordingly, each sample in the whole dataset is ensured to be used exactly once for every epoch iteration. Our limited numerical experiments suggest that this approach often leads to estimators with satisfactory statistical efficiency. As a consequence, shuffled minibatch-based GD algorithms are indeed the most widely used SGD algorithms in practice. In fact, it has been implemented by many important software libraries, such as TensorFlow \citep{abadi2016tensorflow} and PyTorch \citep{paszke2017automatic}.

However, our theoretical understanding of this most popularly used SGD algorithm (i.e., minibatch-based gradient descent with momentum, MGDM) is extremely limited. A number of important theoretical questions remain unknown. For example, under what conditions should an MGDM algorithm numerically converge? Past literature suggests that for a standard GD algorithm, the learning rate should not be too large. Otherwise, the algorithm might not numerically converge. Consequently, it is natural to question whether this remains true for MGDM algorithms. We are particularly interested in understanding the important role played by the momentum parameter. Next, assuming that the MGDM algorithm converges numerically and leads to the MGDM estimator, what are the asymptotic properties of the resulting MGDM estimator? This would facilitate the understanding of the roles played by both the learning rate and momentum parameter.

To address these important questions, we start with a relatively simple model (i.e., a linear regression model) and one particular type of MGDM algorithm (i.e., the fixed minibatch-based GDM algorithm, FMGDM). By doing so, we are able to analyze the theoretical properties of the MGDM method both numerically and asymptotically. This leads to fruitful theoretical findings. In particular, we identify the nearly sufficient and necessary conditions for the tuning parameters (i.e., the learning rate and the momentum parameter) to ensure numerical convergence of the algorithm. In addition, we further investigate the numerical convergence rate and provide the optimal tuning parameters specification in theory.
Subsequently, we study the asymptotic properties of the resulting MGDM estimator.
We find that a diminishing learning rate is needed for the resulting estimator to achieve global statistical efficiency.
Finally, we conduct extensive numerical experiments to validate our theoretical findings.

In summary, we attempt to make the following important contributions to the existing literature. First, we systematically study the MGDM algorithms for both shuffled and fixed minibatch. A relatively complete theory is provided for both numerical and asymptotic properties of the MGDM algorithm. Second, our theory leads to nontrivial findings on the important relationship between statistical efficiency of the resulting estimator and the tuning parameters. In particular, we find that as the learning rate decreases, both MGDM estimators (i.e., fixed and shuffled) can be statistically as efficient as the global OLS estimator.
Last, extensive numerical studies are presented to verify our theoretical findings and demonstrate the outstanding performance of our proposed FMGDM method.

The rest of the article is organized as follows. Section 2 first introduces the MGDM method under the linear regression setup. Then, the numerical convergence properties and the statistical efficiency of the resulting estimator are investigated.
Section 3 includes extensive numerical experiments to corroborate our theoretical findings.
We conclude the article in Section 4 and discuss a number of interesting topics for future study.

\csection{MOMENTUM GRADIENT DESCENT WITH FIXED MINI-BATCH}

\csubsection{Classical Gradient Descent with Momentum}

We first introduce some frequently used notations for this paper. Let $ A $ be a matrix in $ \mR^{p_1 \times p_2} $. We use $ \lambda_{\max}(A) $ and  $ \lambda_{\min}(A) $ to denote the largest and smallest eigenvalue , respectively, of $ A $, if $ A $ is symmetric. We define $ \|A\| =\sqrt{\lambda_{\max}(A^\top A)}$ and $ \|A\|_F =  \sqrt{\tr(A^\top A)}$ as the operator norm and the Frobenius norm of $ A $, respectively. In addition, $ \rho(A) = \max\{|\lambda|: \lambda \text{ is the eigenvalue of }A \} $, which stands for the spectral radius of $ A $. As usual, $ \|v\| = \sqrt{v^\top v} $ is the standard $ l_2 $-norm of an arbitrary vector $ v\in\mR^p $.
We start with a classical linear regression model. Let $ (Y_i, X_i) $ be the observation collected from $ i $-th subject with $ 1\le i\le N $.
Assume
\begin{equation}\label{linear model}
	Y_i = X_i^\top\theta + \varepsilon_i,\  \text{with}\  1\le i\le N ,
\end{equation}
where $ X_i \in\mR^p$ is the covariate, $ Y_i \in \mR$ is the response, $ \varepsilon_i $ is the random error with mean $ 0 $ and variance $ \sigma^2 $, and $ \theta\in\mR^p $ is the unknown parameter vector with true value given by $ \theta_0 = (\theta_{01},\dots, \theta_{0p})^\top$. Suppose the whole sample is indexed by $ \mathbb{S} = \{1,\dots, N\} $. To estimate the unknown parameter vector $ \theta $, we usually minimize the ordinary least squares (OLS) loss function as
\begin{equation*}
	\mL(\theta) =N^{-1}\sum_{i=1}^N \ell_i ( \theta)= (2N)^{-1}\sum_{i=1}^N \Big(Y_i - X_i^\top \theta\Big)^2,
\end{equation*}
where $ \ell_i( \theta) = 2^{-1}(Y_i - X_i^\top \theta)^2 $ is the loss evaluated on the $ i $-th observation. We know that the solution to this problem is the classical ordinary least squares (OLS) estimator that has an explicit form as $ \widehat{\theta}_\text{ols} = \arg\min_{\theta} \mL(\theta) = \widehat{\Sigma}_{xx} ^{-1} \widehat{\Sigma}_{xy}$, where $ \widehat{\Sigma}_{xx} = N^{-1} \sum_{i=1}^N X_i X_i^\top \in\mR^{p\times p} $ and $ \widehat{\Sigma}_{xy} = N^{-1} \sum_{i=1}^N X_i Y_i \in \mR^p$. Under some mild conditions, one can show that $ \sqrt{N} (\widehat{\theta}_\text{ols} - \theta_0) \to_d N(\0, \sigma^2 \Sigma_{xx}^{-1}) $, where $ \Sigma_{xx} $ is the covariance matrix of $ X_i $ \citep{rao1973linear}. To practically compute $ \widehat{\theta}_\text{ols} $, the sample covariance matrix $ \widehat{\Sigma}_{xx} $ needs to be calculated and inverted. This leads to a computational complexity of order $ O(Np^2 + p^3) $ in general.

Another popular way to compute $ \widehat{\theta}_\text{ols} $ is to use the first-order optimization methods in an iterative way.
Arguably, the most basic first-order optimization method is the gradient descent (GD) algorithm.
For the OLS problem, a standard GD algorithm should update the estimates iteratively as $\widehat \theta^{(t)} = \widehat \theta ^{(t-1)} - \alpha \dot{\mL}(\widehat \theta ^{(t-1)})$, where $  \dot{\mL}(\theta ) $ denotes the gradient of $ \mL $ at $ \theta $ and $ \alpha > 0 $ is the learning rate controlling the step size.
As shown by \cite{nesterov2018lectures}, with an appropriately selected $\alpha $, the estimates $ \{\widehat{\theta}^{(t)} \} $ generated by the GD algorithm should satisfy $ \|\widehat{\theta}^{(t)} -  \widehat{\theta}_\text{ols}\| \le \rho^t  \|\widehat{\theta}^{(0)} -  \widehat{\theta}_\text{ols}\|$ for some convergence factor $ \rho \in [0,1) $ and an arbitrary initial value $ \widehat{\theta}^{(0)} $.
The optimal convergence factor can be achieved if we choose $ \alpha = 2 / \{\lambda_{\max}(\widehat{\Sigma}_{xx})+ \lambda_{\min}(\widehat{\Sigma}_{xx}) \} $.
In this case, the optimal convergence factor should be $ \rho = (\kappa - 1)/ (\kappa +1) $ , where $ \kappa = \lambda_{\max}(\widehat{\Sigma}_{xx}) / \lambda_{\min}(\widehat{\Sigma}_{xx}) \ge 1$ is the condition number of $ \widehat{\Sigma}_{xx} $.

However, when the condition number $ \kappa $ is large, the convergence factor $ \rho $ could be very close to $ 1 $.
Consequently, the convergence rate could be painfully slow.
To accelerate it, \cite{polyak1964} proposed the heavy-ball method, which is also called the momentum method \citep{sutskever2013importance}.
Specifically, in this work, we investigate the gradient descent with the classical momentum, which updates the estimates as
\begin{gather}
	\widehat{\theta}^{(t)} = \widehat{\theta}^{(t-1)} - v^{(t)}\label{MGD1},\\
 	v^{(t)} = \gamma v^{(t-1)} + \alpha \dot{\mL}(\widehat \theta ^{(t-1)})\label{MGD2},
\end{gather}
where $  \alpha >0$ is the learning rate and $\gamma>0 $ is the momentum parameter. \cite{polyak1964} showed that, by choosing appropriate tuning parameters $ \alpha$ and $ \gamma $, the optimal convergence factor of the momentum method could be achieved as $ \rho = (\sqrt{\kappa} - 1) / (\sqrt{\kappa} + 1) $. In contrast, that of a standard GD algorithm is $ \rho = (\kappa - 1)/ (\kappa +1)  $, which is strictly larger than that of the GD with momentum (GDM) method as long as $ \kappa>1 $. This suggests that the optimal convergence rate of a GDM algorithm should be faster than that of a standard GD algorithm.

\csubsection{ A Fixed Mini-batch Based Gradient Descent with Momentum}

Despite its theoretical attractiveness, a standard GDM method cannot be immediately used to handle massive datasets.
In many cases, the dataset could be too large to be comfortably processed as a whole by computer memory.
As a consequence, it has to be placed on the hard drive and then processed in a batch-by-batch manner.
In practice, this approach leads to the most popularly used minibatch-based gradient descent methods with momentum (MGDM).
Specifically, to practically implement an MGDM algorithm, we first need to randomly partition the whole sample $\mathbb{S}$  into $ M $ disjoint minibatches as $ \mathbb{S}=\cup_{m=1}^M  \mS_{(m)}$ with $ \mS_{(m_1)} \cap \mS_{(m_2)} = \emptyset $ for any $ m_1 \ne m_2 $.
Throughout the article, we assume that $ M $ is fixed.
Without loss of generality, we assume that $ |\mS_{(m)} |=n $ for each $ 1\le m\le M $.
Obviously, we should have $ N = M n $.
As mentioned before, we start with the fixed minibatch-based GDM (FMGDM) method due to its analytical simplicity.
In other words, we assume that the random partition $ \mathbb{S}=\cup_{m=1}^M  \mS_{(m)}  $, once given, should be fixed throughout the rest of the algorithm.

In the fixed minibatch setting, the OLS loss function based on $ \mS_{(m)} $ can be defined as $ \mL_{(m)}(\theta)=(2n)^{-1}\sum_{i \in \mS_{(m)}} (Y_i - X_i^\top \theta)^2 $.
The corresponding gradient is $ \dot \mL_{(m)}(\theta) = \widehat \Sigma_{xx}^{(m)}\theta - \widehat \Sigma_{xy}^{(m)} $, where $ \widehat{\Sigma}_{xx}^{(m)} = n^{-1} \sum_{i \in \mS_{(m)} }X_i X_i^\top\in \mR^{p\times p}$ and $ \widehat{\Sigma}_{xy}^{(m)} = n^{-1} \sum_{i \in \mS_{(m)} }X_i Y_i\in \mR^p$.
Subsequently, a standard GDM algorithm should be executed in a batch-by-batch manner.
Specifically, let $ \widehat{\theta}^{(t, m)} $ be the estimate obtained from the $ m $-th minibatch in the $ t $-th epoch.
We then should have
\begin{gather}
	\widehat \theta^{(t, m)} = \widehat \theta ^{(t, m-1 )} - v^{(t,m)}, \label{FMBM1}\\
	v^{(t,m)} = \gamma v^{( t,m-1)} + \alpha \dot\mL_{(m)}(\widehat \theta^{(t, m-1)})\label{FMBM2},
\end{gather}
where $ \widehat \theta^{(t,0)} =   \widehat \theta^{(t-1,M)}$, $ v^{(t,0)}=v^{(t-1,M)} $, and $ \widehat \theta^{(0,m)} $ ($ 1\le m\le M $) is the initial estimate specified for the $ m $-th minibatch.
This immediately leads to three important questions.
First, theoretically, if we can treat \eqref{FMBM1} and \eqref{FMBM2} as a linear dynamical system, does this dynamical system have a stable solution?
Second, does the FMGDM algorithm converge to this stable solution?
Last, if the FMGDM algorithm does converge to the stable solution under appropriate conditions, what is the statistical properties of the resulting estimator?

\csubsection{The Linear Dynamical System and Stable Solution}

We investigate the first question in this subsection: does there exist a stable solution to the FMGDM algorithm?
To address this important problem, we temporarily assume that there does exist a numerical limit $ \widehat \theta^{(m)} $ such that $ \widehat \theta^{(t, m)}\to\widehat \theta^{(m)} $ as $t\to\infty $ for each $ 1\le m\le M $.
This leads to an analytical solution as $ \widehat \theta^{(m)}  $ for each $ 1\le m\le M $.
This analytical solution further helps us to understand the conditions needed for its unique existence.
Specifically, we write the stable solution vector as $ \widehat \theta^* = (\widehat \theta^{(1)\top}, \dots, \widehat \theta^{(M)\top})^\top \in \mR^{q} $ with $ q = Mp $.
Then, by \eqref{FMBM1} and \eqref{FMBM2}, we know that $  \widehat \theta^* $ should be the solution to the following linear dynamical system,
\begin{align} \label{dynamical system}
\begin{cases}
	\widehat \theta^{(t,1)} = \wh\Delta^{(1)} \widehat \theta^{(t,M)} - \gamma \widehat \theta^{(t,M-1)} + \alpha\widehat \Sigma_{xy}^{(1)}\\
	\widehat \theta^{(t,2)} = \wh\Delta^{(2)}  \widehat \theta^{(t,1)} - \gamma \widehat \theta^{(t,M)} + \alpha\widehat \Sigma_{xy}^{(2)}\\
	\cdots\\
	\widehat \theta^{(t,M)} = \wh\Delta^{(M)}  \widehat \theta^{(t,M-1)} - \gamma \widehat \theta^{(t,M-2)} + \alpha\widehat \Sigma_{xy}^{(t,M)},\\
\end{cases}
\end{align}
where $ \wh\Delta^{(m)} =(1+\gamma) I_p - \alpha \widehat \Sigma_{xx}^{(m)} \in \mR^{p\times p}$ for $ 1\le m\le M $.
We further write $ \widehat \Sigma_{xy}^* =   \big(\widehat\Sigma_{xy}^{(1)\top}, \dots, \widehat \Sigma_{xy}^{(M)\top}\big)^\top \in \mR^{q}$, and
\begin{align}\label{Omega_hat}
\wh\Omega = \begin{bmatrix}
	I_p&\0&\0&\cdots&\0&\gamma I_p&-\wh\Delta^{(1)}\\
	-\wh\Delta^{(2)}&I_p&\0&\cdots&\0&\0&\gamma I_p\\
	\gamma I_p&-\wh\Delta^{(3)}&I_p&\cdots&\0&\0&\0\\
	\vdots&\vdots&\vdots&\vdots&\vdots&\vdots &\vdots\\
	\0&\0&\0&\cdots&\gamma I_p&-\wh\Delta^{(M)}&I_p
\end{bmatrix}\in \mR^{q\times q}.
\end{align}
Then, we should have $ \wh\Omega \widehat \theta^* = \alpha\widehat \Sigma_{xy}^* $ or equivalently $\widehat \theta^* = \alpha \wh\Omega^{-1} \widehat \Sigma_{xy}^*  $, provided $ \wh\Omega $ is invertible.
Consequently, whether $ \wh\Omega $ is invertible or not determines the existence of the stable solution.
Then, in the following theorem, we elaborate the conditions to ensure its invertibility in probability.

\begin{theorem} {\sc (Stable Solution)} \label{thm:invertible}
Assume $ \Sigma_{xx}\in \mR^{p\times p} $ is a positive definite matrix with eigenvalues $ \lambda_1 \ge \lambda_2\ge \cdots \ge \lambda_p >0 $.
Further assume that $ \gamma\ne 1 $, $ \alpha \ne 0 $, and $ \alpha \ne 2(1+\gamma)/\lambda_j $ for each $ 1\le j \le p $.
Then, $ \wh\Omega$ is invertible with probability tending to one as $ n\to \infty $.
\end{theorem}

The proof of Theorem \ref{thm:invertible} can be found in Appendix A.1.
By Theorem \ref{thm:invertible}, we know that the stable solution does exist for a wide range of different $ \alpha $ and $ \gamma $ values.
In fact, as long as  $ \gamma\ne 1 $, $ \alpha \ne 0 $, and $ \gamma \ne 2(1+\gamma)/ \lambda_j $ for every $ 1\le j\le p  $, we should have $ \wh\Omega $ invertible with high probability, and thus, the stable solution should exist.
In particular, we find that the stable solution exists as $ \gamma\to \infty $.
Nevertheless, we should remark that the existence of the stable solution and the FMGDM estimator might be totally different.
By the FMGDM estimator, we mean that the FMGDM algorithm should converge to a numerical limit $ \widehat{\theta}^{(m)} $, and this numerical limit $ \widehat{\theta}^{(m)} $ is referred to as the FMGDM estimator.
Obviously, as long as such a numerical limit does exist, it must be a stable solution to \eqref{dynamical system}.
Nevertheless, this does not imply that every possible stable solution can be numerically approached by the FMGDM algorithm.
In fact, there do exist stable solutions (e.g., those stable solutions associated with very large $ \gamma $ values) that cannot be achieved using the FMGDM algorithm.
Consequently, it is not a FMGDM estimator.
Thus, we know that the existence of the stable solution should be a necessary condition for the existence of the FMGDM estimator but might not be sufficient.
Accordingly, we investigate the conditions to assure the numerical convergence of the FMGDM algorithm in the next subsection.

\csubsection{Numerical Convergence Properties}

We devote this subsection to studying the second important problem, that is, the numerical convergence properties of the FMGDM algorithm.
Specifically, we are interested in studying how the $ \alpha $ and $ \gamma $ would affect the numerical convergence of the FMGDM algorithm.
Recall that $ \lambda_j\ (1\le j \le p)$ is the $ j $-th largest eigenvalue of $\Sigma_{xx} $.
The main results are summarized in the following theorem.

\begin{theorem}\label{thm:convergence}
{ \sc (Numerical Convergence)}
(a) If $ 0\le \gamma <1 $ and $0<\alpha< 2(1+\gamma)/\lambda_{1}$ and assuming $ n\to \infty $, then, we have $ \widehat{\theta}^{(t, m)} \to \widehat{\theta}^{(m)}$ as $ t\to \infty $ holds with probability tending to one.
(b) Assume $ \widehat{\theta}^{(0, m)} \neq   \widehat{\theta}^{(m)}$ and $ n\to \infty $. For $ \gamma>1 $ or $ \alpha> 2(1+\gamma)/\lambda_1 $, we have $ \widehat{\theta}^{(t, m)} \not\to \widehat{\theta}^{(m)}$ as $ t\to \infty $ holds with probability tending to one.
\end{theorem}

The proof of Theorem \ref{thm:convergence} can be found in Appendix A.2. Theorem \ref{thm:convergence} provides nearly sufficient and necessary conditions for numerical convergence.
Specifically, from Theorem \ref{thm:convergence}(a), we know that to ensure the numerical convergence, we must have $0 \le \gamma <1$.
 This corroborates the common practice in the existing literature very well \citep{sutskever2013importance}.
Otherwise, by Theorem \ref{thm:convergence}(b), we know that the numerical convergence can never be achieved if $ \gamma>1 $.
Consequently, the FMGDM estimator does not exist in this case.
Furthermore, with a given momentum parameter $0\le\gamma<1$, the learning rate $ \alpha $ should be sufficiently small.
Otherwise, the algorithm should not converge.
Remarkably, the upper bound of the learning rate $ \alpha $ is related to the largest eigenvalue of $ \Sigma_{xx} $ (i.e., $ \lambda_{1} $) and the momentum parameter $ \gamma $.
This upper bound should be larger if the momentum parameter $ \gamma $ is larger.
As a result, this enables a larger learning rate (thus faster convergence) to be used for a FMGDM algorithm compared to a standard GD algorithm ($ \gamma=0 $).
Theorem \ref{thm:convergence} examines the conditions for the numerical convergence for the FMGDM algorithm.
However, it does not provide the exact convergence rate.
In particular, we want to understand how the $ \alpha $ and $ \gamma $ specification would affect the numerical convergence rate.
We next develop the following theorem.

\bet\label{thm:convergence_speed}
{\sc (Numerical Convergence Speed)}
If $ 0\le \gamma <1 $ and $0<\alpha< 2(1+\gamma)/\lambda_{1}$,
then there exists a constant $ \rho \in[\gamma,1) $, such that $ \|\widehat{\theta}^{(t, m)} - \widehat{\theta}^{(m)} \| \le (\rho^M + \zeta_n +\eta_t)^t \big(\|\widehat{\theta}^{(0, m)} - \widehat{\theta}^{(m)} \| + \|\widehat{\theta}^{(0, m-1)} - \widehat{\theta}^{(m-1)} \|\big) $ for each $ 1\le m\le M $, where $ \zeta_n \to_p 0 $ as $ n \to \infty $ and $ \eta_t \to 0 $ as $ t\to \infty $.
Furthermore, we have
\begin{enumerate}[(a)]

\item The minimal (i.e., optimal) $ \rho $ is given by $ \rho = (\sqrt{\lambda_{1}/\lambda_{p}} - 1) / (\sqrt{\lambda_{1}/\lambda_{p}} + 1)$, when $ \alpha = 4/(\sqrt{\lambda_{1} }+ \sqrt{\lambda_p})^2 $ and $ \gamma = (\sqrt{\lambda_{1}} - \sqrt{\lambda_p})^2 / (\sqrt{\lambda_{1}} + \sqrt{\lambda_p})^2 $.

\item If $ \gamma=0 $, we should have $ \rho =   \max\{|1-\alpha\lambda_1|,  |1-\alpha\lambda_p| \} $.
In particular, the minimal (i.e., optimal) $ \rho $ is given by $ \rho = (\lambda_1/\lambda_{p} -1 ) /  (\lambda_1/\lambda_{p} +1)$, when $ \alpha = 2/ (\lambda_{1}+\lambda_p) $.

\item Assume $ 0<\alpha< 1/\lambda_1$.
(c.1) If $ \gamma = 0 $, then the $ \rho $ is given by $\rho =  1-\alpha \lambda_p $.
(c.2) If $0<  \gamma <(1-{\alpha \lambda_p})^2 $, then $\rho < 1-\alpha \lambda_p$.
In particular, the minimal (i.e., optimal) $ \rho $ is given by $\rho = 1-\sqrt{\alpha \lambda_p} $, when $ \gamma = (1-\sqrt{\alpha \lambda_p})^2$.


\end{enumerate}

\eet

The proof of Theorem \ref{thm:convergence_speed} is shown in Appendix A.3.
From Theorem \ref{thm:convergence_speed}, we know that the FMGDM algorithm should enjoy a linear convergence rate by choosing appropriate $ \alpha $ and $ \gamma $, as long as both $ n $ and $ t $ are large enough.
By Theorem \ref{thm:convergence_speed} (a), we know that the optimal convergence factor $ \rho =  (\sqrt{\lambda_{1}/\lambda_{p}} - 1) / (\sqrt{\lambda_{1}/\lambda_{p}} + 1)$ can be achieved by carefully specified $ \alpha $ and $ \gamma $.
This specification is closely related to the eigenvalue structure of the covariance matrix $ \Sigma_{xx} $.
Our findings are in line with theoretical findings in the classical GDM method without minibatch \citep{polyak1964}. It is known that this convergence factor also attains the  achievable lower bound of various first-order methods \citep{nesterov2018lectures}. 
Moreover, Theorem \ref{thm:convergence_speed} (b) claims that,  for a standard GD algorithm with $ \gamma=0 $, the best convergence factor should be $ \rho  = (\lambda_1/\lambda_{p} -1 ) /  (\lambda_1/\lambda_{p} +1)$.
This represents a convergence rate that is strictly slower than that of the FMGDM algorithm as long as $ \lambda_1>\lambda_p $.

Although the optimal numerical convergence can be guaranteed by (a) and (b) under different circumstances, the stable solution is not necessarily optimal in statistical efficiency.
For most SGD-related methods, the decaying learning rate $\alpha$ is typically needed to improve statistical efficiency \citep{robbins1951stochastic, polyak1992acceleration}.
In this case, from Theorem \ref{thm:convergence_speed} (c.1) we know that the convergence factor in this case tends to $ 1 $ as $ \alpha \to 0$, which represents a very slow convergence rate.
Moreover, from Theorem \ref{thm:convergence_speed} (c.2) we find that by choosing an appropriate $ \gamma $, the convergence factor of the FMGDM could be strictly smaller than that of a standard GD algorithm ($ \gamma=0 $).
This finding suggests that the momentum term can indeed accelerate the numerical convergence.
In addition, Theorem \ref{thm:convergence_speed} (c.2) discovers a very interesting and novel interactive relationship between the learning rate $ \alpha $ and the momentum parameter $ \gamma $. Specifically, it suggests that $ \gamma $ should approach $ 1 $ as $ \alpha\to 0 $. This result theoretically explains why practitioners often like to specify $ \gamma $ to be a value very close to $ 1 $ to strike a balance between both numerical and statistical convergence \citep{sutskever2013importance}.

\csubsection{The Statistical Properties of the Stable Solution}

In this subsection, we investigate how the tuning parameters $ \alpha$ and $ \gamma $ would affect the statistical properties of the stable solution.
Note that the OLS estimator is the global optimal solution to our problem.
Ideally, we should have the resulting FMGDM estimator stay with the OLS estimator as closely as possible.
Consequently, we should examine the relationship between the stable solution and the OLS estimator.
To this end, we need to analyze $ \wh\Omega $ and its inverse (assume it exists) carefully.
Accordingly, we decompose it into $ \wh\Omega = A + \alpha \wh B $, where
\begin{align}
A = \begin{bmatrix}
	I_p&\0&\0&\cdots&\0&\gamma I_p&-(1+\gamma)I_p\\
	-(1+\gamma)I_p&I_p&\0&\cdots&\0&\0&\gamma I_p\\
	\gamma I_p&-(1+\gamma)I_p&I_p&\cdots&\0&\0&\0\\
	\vdots&\vdots&\vdots&\vdots&\vdots &\vdots&\vdots\\
	\0&\0&\0&\cdots&\gamma I_p&-(1+\gamma)I_p&I_p
\end{bmatrix}\in \mR^{q\times q} \label{A},
\end{align}
and
\begin{align}
\wh B = \begin{bmatrix}
	\0&\0&\cdots&\0&\0& \widehat\Sigma_{xx}^{(1)}\\
	\widehat\Sigma_{xx}^{(2)}&\0&\cdots&\0&\0&\0\\
	\vdots&\vdots&\vdots&\vdots&\vdots &\vdots\\
	\0&\0&\cdots&\0& \widehat\Sigma_{xx}^{(M)}&\0
\end{bmatrix}\in \mR^{q \times q} \label{B}.
\end{align}
Note that $ A $ and $ \wh B $ do not involve the learning rate $ \alpha $.
Furthermore, it can be easily verified that $ A^\top I^* = AI^* = \0 $, where $ I^* = \boldsymbol{1}_M\otimes I_p $.
Thus, the columns of $ I^* $ are the eigenvectors of $ A $ and $ A^\top $ corresponding to the eigenvalue $ 0 $. Next, let $ P_1 = I^* / \sqrt{M} \in \mR^{q\times p}$.
Then, we can construct $ P_2 \in \mR^{q\times (q-p)}$ that has orthonormal columns and is orthogonal to $ P_1 $ (i.e., $ P_2^\top P_2 = I_p,\ P_2^\top P_1 = \0 $).
We then have $ P = [P_1, P_2] \in \mR^{q \times q}$ as an orthogonal matrix.
Simple algebra gives $ P^\top A P = [\0, \0; \0, P_2^\top A P_2] $.
Consequently, we obtain that
\begin{align}\label{P Omega P}
P^\top \wh\Omega P = P^\top A P + \alpha P^\top \wh B P = \begin{bmatrix}
	\alpha P_1^\top \wh B P_1&\alpha P_1^\top \wh B P_2\\
	\alpha P_2^\top \wh B P_1&P_2^\top A P_2 + \alpha P_2^\top \wh B P_2
\end{bmatrix}.
\end{align}
Denote $\wh B_{11}^P= P_1 ^\top \wh B P_1  $, $\wh B_{12}^P= P_1 ^\top \wh B P_2  $, $\wh B_{21}^P= P_2 ^\top \wh B P_1  $, $\wh B_{22}^P= P_2 ^\top \wh B P_2  $, and $A_{22}^P= P_2 ^\top A P_2  $. Then, \eqref{P Omega P} is simplified to $ P^\top \wh\Omega P = [\alpha \wh B_{11}^P, \alpha \wh B_{12}^P;\, \alpha \wh B_{21}^P, A_{22}^P + \alpha \wh B_{22}^P] $. With the help of this transformation, we can obtain the following theorem, which elaborates the difference between the stable solution and the OLS estimator.

\begin{theorem} \label{thm:ss}
	{\sc (Relationship with the OLS Estimator)}
	Suppose the conditions in Theorem \ref{thm:invertible} hold.
	Further denote $ \widehat{\theta}_\text{ols}^* = \1_M\otimes  \widehat{\theta}_\text{ols} $. Then we have
	\begin{enumerate}[(a)]
		\item If $ \gamma\ge 0 $ is fixed and $ \alpha \to 0$, then $ \widehat{\theta}^*  = \widehat{\theta}_\text{ols}^*+\alpha \widehat{E} \{1+O_p(\alpha)\}$, where $ \widehat{E} =  \big(P_1 \widehat{\Sigma}_{xx} ^{-1}P_1^\top \wh B - I_q \big)  \big\{P_2 (A_{22}^P)^{-1} P_2^\top \big\} \big(\wh B \widehat{\theta}_\text{ols}^*-\widehat{\Sigma}_{xy}^*  \big)$.
		\item If $ \alpha>0 $ is fixed and $ \gamma \to \infty $ , then $ \widehat{\theta}^* = \widehat{\theta}_\text{ols}^*+\alpha \widehat{E} \{1+O_p( \|(A_{22}^P)^{-1}\|  )\}$, where $ \widehat{E} $ is the same as that in (a) and $ \|(A_{22}^P)^{-1}\|\to 0  $ as $ \gamma \to \infty $.
	\end{enumerate}

\end{theorem}

The proof of Theorem \ref{thm:ss} can be found in Appendix A.4. With a fixed $ \gamma $, it can be verified that $ \sqrt{n}\wh E = O_p(1) $  as $ n\to \infty $. Then by Theorem \ref{thm:ss}, we should have $\alpha\sqrt n \widehat{E} \to_p 0 $ if $ \alpha\to 0 $.
Hence it can be ignored since we have $\sqrt n (\widehat \theta_\text{ols} - \theta_0) = O_p(1)$.
With a fixed $ \alpha $, it can be verified that $ (A_{22}^P)^{-1}  \to 0$ as $\gamma \to \infty $.
This further implies that $ \sqrt n \widehat{E}\to_p 0 $ as $ \gamma \to \infty $.
Consequently, we should have $ \widehat{\theta}^{(m)} \to \widehat{\theta}_\text{ols} $ for $ 1\le m \le M $ as long as $ \alpha\to 0 $ or $ \gamma \to \infty $.
This immediately suggests two different ways to force the stable solution to stay close to the global OLS estimator.
They are, respectively, $ \alpha\to 0 $ and $ \gamma\to\infty $.
However, from Theorem \ref{thm:convergence} (b), we know that the FMGDM algorithm does not converge if $ \gamma>1 $.
Consequently, this leaves us only one way to obtain the best statistical efficiency, which is forcing $ \alpha\to0 $.
On the other hand, if $ \alpha $ is set to be too small, then the FMGDM algorithm would converge painfully slowly.
Thus, practitioners have to consider relatively large $ \alpha $ values at the early stage of the algorithm for faster numerical convergence.
Then we turn to smaller values for improved statistical efficiency.
Moreover, as suggested by Theorem \ref{thm:convergence_speed} (c.2), we can choose a momentum parameter close to $ 1 $ to accelerate the numerical convergence.

To further investigate the statistical property of the stable solution, we establish its asymptotic normality for fixed $ \gamma $ in the following theorem.

\begin{theorem} \label{thm:asy_normality}
	{\sc (Asymptotic Normality)}
Suppose the conditions in Theorem \ref{thm:invertible} hold.
Further assume that $ \gamma\ge 0 $ is fixed and $ \alpha \to 0 $. This yields
	\begin{equation*}
		\sqrt{N} \big(\wh \theta^{(m)} - \theta_0 + \wh b_m(\alpha) \big) \to_d N_p\Big(\0, \sigma^2 V_m(\alpha) \Big)
	\end{equation*} as $ N\to \infty $, where $ \wh b_m(\alpha) = O_p(\alpha^2 n^{-3/2} ) $, $ V_m(\alpha) = \Sigma_{xx}^{-1} +\alpha^2 M d_\gamma^2 \Sigma_{xx} +O(\alpha^4)$, and $ d_\gamma $ is the $ l_2 $-norm of the first row of the matrix $ P_2 (A_{22}^P)^{-1} P_2^\top $.
\end{theorem}
The proof of Theorem \ref{thm:asy_normality} is shown in Appendix A.5 and the explicit formulas for calculating $ d_\gamma $ can be found in Appendix A.6. Theorem \ref{thm:asy_normality} gives the asymptotic behavior of the stable solution. Specifically, for fixed $ \gamma $ and $ M $, the stable solution is still $ \sqrt{N} $-consistent
as long as $\sqrt N\wh b_m(\alpha) = o_p(1)$.
In addition, from the bias and variance term, we can see that as $ \alpha \to 0 $, the stable solution should be asymptotically equivalent to the OLS estimator.

\csubsection{Shuffled Minibatch-Based Gradient Descent with Momentum}

As mentioned before, the shuffled minibatch-based GDM (SMDGM) algorithm is in fact the most popularly used MGDM algorithm in practice, but little is understood regarding its statistical properties. However, Theorem \ref{thm:ss} developed in the previous subsection provides us an opportunity to obtain some theoretical insights.
Specifically, let $ \mS^{(k,m)} $ be the $ m $-th minibatch used in the $ k $-th shuffled partition ($ 1\le k\le K $). We should have $ \mathbb{S} = \cup_{m=1}^M  \mS^{(k,m)} $ for every $ k $ and $ \mS^{(k,m_1)} \cap \mS^{(k,m_2)} = \emptyset   $ for any $ m_1 \ne m_2 $ and every $ k $.
Different shuffled partitions subsequently determine different dynamical systems and thus different stable solutions.
More specifically, consider one particular shuffled partition $\{ \mS^{(k,m)} : 1\le m\le M \}  $ for one particular $1\le k\le K $. Then, the corresponding dynamical system becomes
\begin{align} \label{dynamical system shuffled}
	\begin{cases}
		\widehat \theta^{(t, 1)} = \wh\Delta^{(k, 1)} \widehat \theta^{(t, M)} - \gamma \widehat \theta^{(t, M-1)} + \alpha\widehat \Sigma_{xy}^{(k, 1)}\\
		\widehat \theta^{(t,2)} = \wh\Delta^{(k,2)}  \widehat \theta^{(t,1)} - \gamma \widehat \theta^{(t,M)} + \alpha\widehat \Sigma_{xy}^{(k,2)}\\
		\cdots\\
		\widehat \theta^{(t,M)} = \wh\Delta^{(k,M)}  \widehat \theta^{(t, M-1)} - \gamma \widehat \theta^{(t,M-2)} + \alpha\widehat \Sigma_{xy}^{(k,M)},\\
	\end{cases}
\end{align}
where $ \wh\Delta^{(k, m)} =(1+\gamma) I_p - \alpha \widehat \Sigma_{xx}^{(k, m)} \in \mR^{p\times p}$, $ \wh\Sigma_{xx}^{(k, m)}  = n^{-1}\sum_{i\in\mS^{(k, m)}} X_i X_i^\top $, and $  \wh\Sigma_{xy}^{(k, m)}  = n^{-1}\sum_{i\in\mS^{(k, m)}} X_i Y_i  $ for each $ 1\le m\le M $.
We denote the corresponding stable solution by $ \wh \theta^{(k)*} $. Then, from Theorem \ref{thm:ss}, we know that $  \wh \theta^{(k)*}  - \wh\theta_\text{ols}^* = \alpha \wh E^{(k)} \{1+O_p( \alpha  )\} $, where $ \wh E^{(k)} =  \big(P_1 \widehat{\Sigma}_{xx} ^{-1}P_1^\top \wh B^{(k)} - I_q \big)  \big\{P_2 (A_{22}^P)^{-1} P_2^\top \big\} \big(\wh B^{(k)}  \widehat{\theta}_\text{ols}^* -\widehat{\Sigma}_{xy}^{(k)*}  \big) $, $ \wh B^{(k)}  $ has the same form as \eqref{B} but with $ \wh\Sigma_{xx}^{( m)} $ replaced by $ \wh\Sigma_{xx}^{(k, m)} $ for each $ 1\le m\le M $, and $ \wh\Sigma_{xy}^{(k)*} =  \big(\wh\Sigma_{xy}^{(k, 1)\top}, \dots, \wh \Sigma_{xy}^{(k, M)\top}\big)^\top $.
Consequently, as long as we can show that $ \max_{1\le k \le K}\|\wh E^{(k)}\|=o_p( 1/\sqrt{N}) $, then the difference between the shuffled estimator $  \wh \theta^{(k)*} $ and the global OLS estimator is ignorable uniformly over all shuffled partitions.

Specifically, we first note that  $\max_{1\le k \le K} \| \wh E^{(k)}\| \le \big(\max_{1\le k\le K} \big\|P_1 \widehat{\Sigma}_{xx} ^{-1}P_1^\top \wh B^{(k)} \big\|+ 1 \big)\  \big\|P_2 (A_{22}^P)^{-1} P_2^\top    \big\| \max_{1\le k\le K} \big\|\wh B^{(k)}  \widehat{\theta}_\text{ols}^* -\widehat{\Sigma}_{xy}^{(k)*}  \big\|$, where
\begin{align*}
	P_1 \widehat{\Sigma}_{xx} ^{-1}P_1^\top \wh B^{(k)} = M^{-1}\begin{bmatrix}
		\wh\Sigma_{xx}^{-1} \wh\Sigma_{xx}^{(k, 2)} &\wh\Sigma_{xx}^{-1} \wh\Sigma_{xx}^{(k, 3)}&\cdots&\wh\Sigma_{xx}^{-1} \wh\Sigma_{xx}^{(k, 1)}\\
		\wh\Sigma_{xx}^{-1} \wh\Sigma_{xx}^{(k, 2)} &\wh\Sigma_{xx}^{-1} \wh\Sigma_{xx}^{(k, 3)}&\cdots&\wh\Sigma_{xx}^{-1} \wh\Sigma_{xx}^{(k, 1)}\\
		\vdots&\vdots&\vdots&\vdots\\
		\wh\Sigma_{xx}^{-1} \wh\Sigma_{xx}^{(k, 2)} &\wh\Sigma_{xx}^{-1} \wh\Sigma_{xx}^{(k, 3)}&\cdots&\wh\Sigma_{xx}^{-1} \wh\Sigma_{xx}^{(k, 1)}\\
	\end{bmatrix},
\end{align*}
and $ \|\wh B^{(k)}  \widehat{\theta}_\text{ols}^* -\widehat{\Sigma}_{xy}^{(k)*}  \|= \sum_{m=1}^{M}  \|  \wh\Sigma_{xx}^{(k, m)}\wh \theta_\text{ols} -  \wh\Sigma_{xy}^{(k, m)} \|$. For simplicity, we next assume that both the covariate $ X_i  $ and the random error $ \varepsilon_i $ are normally distributed.
We are then able to study the asymptotic behavior of $ \max_{1\le k \le K} \| \wh E^{(k)}\| $.
To this end, we need to study $\max_{1\le k\le K} \|P_1 \widehat{\Sigma}_{xx} ^{-1}P_1^\top \wh B^{(k)} \| $ and $  \max_{1\le k\le K} \big\|\wh B^{(k)}  \widehat{\theta}_\text{ols}^* -\widehat{\Sigma}_{xy}^{(k)*} \| $ separately.
We first consider  $\max_{1\le k\le K} \|P_1 \widehat{\Sigma}_{xx} ^{-1}P_1^\top \wh B^{(k)} \| $.
By Lemma A.3 in \cite{bickel2008regularized}, there should exist constants $ C_1>0 $ and $ C_2>0 $ such that $ P(\max_{1\le m\le M}\|\wh\Sigma_{xx}^{(k, m)} - \Sigma_{xx}\|_F \ge\nu )\le C_1\exp(-C_2 N\nu^2) $ for any sufficiently small $ \nu>0 $, as long as $ M $ and $ p $ are fixed. Consequently, we should have $ P(\max_{m,k}\|\wh\Sigma_{xx}^{(k, m)} - \Sigma_{xx}\|_F \ge\nu )\le C_1\exp(\log K-C_2 N\nu^2)  $, which tends to $ 0 $ as long as $ \log K = o( N) $.
Therefore, we can derive that $ \max_{m,k}\|\widehat{\Sigma}_{xx} ^{-1} \wh\Sigma_{xx}^{(k, m)}\|  \le \|\widehat{\Sigma}_{xx} ^{-1}\|\max_{m,k}\| \wh\Sigma_{xx}^{(k, m)}\|_F  =O_p(1) $. As a consequence, we should have $ \max_{1\le k\le K} \|P_1 \widehat{\Sigma}_{xx} ^{-1}P_1^\top \wh B^{(k)} \|=O_p(1) $ as long as $ \log K = o(N) $.

Next, we study $  \max_{1\le k\le K} \big\|\wh B^{(k)}  \widehat{\theta}_\text{ols}^* -\widehat{\Sigma}_{xy}^{(k)*} \| $. Let $ \wh \Sigma_{x\varepsilon} = N^{-1} \sum_{i=1}^{N} X_i\varepsilon_i$, and $ \wh \Sigma_{x\varepsilon}^{(k,m)}  =  n^{-1}\sum_{i\in\mS^{(k, m)}} X_i \varepsilon_i $.
Then, we should have $ \widehat{\Sigma}_{xy}^{(k,m)}  = \widehat{\Sigma}_{xx}^{(k,m)} \theta_0 + \Sigma_{x\varepsilon}^{(k,m)} $, and $\wh \theta_\text{ols} =\theta_0 + \widehat{\Sigma}_{xx}^{-1}  \wh \Sigma_{x\varepsilon}$.
Note that $ \|  \wh\Sigma_{xx}^{(k, m)}\wh \theta_\text{ols} -  \wh\Sigma_{xy}^{(k, m)} \|  = \|    \wh\Sigma_{xx}^{(k, m)}  \wh\Sigma_{xx}^{-1}  \wh\Sigma_{x\varepsilon} -  \wh\Sigma_{x\varepsilon}^{(k, m)} \| \le \| \wh\Sigma_{xx}^{(k, m)} \| \| \wh\Sigma_{xx}^{-1}  \wh\Sigma_{x\varepsilon}  \| + \|\wh\Sigma_{x\varepsilon}^{(k, m)} \|$.
Then by  Bernstein’s inequality \citep[Theorem 2.8.2]{vershynin2018high}, we have $ P( \sqrt{N / \log K}\max_{1\le m\le M} \|\wh\Sigma_{x\varepsilon}^{(k, m)} \| \ge\nu)\le  C_1 \exp\big[-C_2 \min\{(\log K) \nu^2, \sqrt{N\log K} \nu\}\big]$ for any $ \nu> 0 $ as long as $ p $ and $ M $ are fixed, where $ C_1, C_2 $ are two positive constants. Hence, it follows that $ P(\sqrt{N / \log K} \max_{m,k}  \|\wh\Sigma_{x\varepsilon}^{(k, m)} \| \ge\nu)\le C_1 \exp\big[\log K-C_2 \min\{(\log K) \nu^2, \sqrt{N\log K} \nu\}\big]$ for any $ \nu> 0 $. As a consequence, for any $ \delta>0 $, there should be a sufficiently large constant $ \nu_\delta>0$ such that $ P(\sqrt{N/\log K} \max_{m,k}  \|\wh\Sigma_{x\varepsilon}^{(k, m)} \| \ge\nu_\delta)\le\delta  $ for each $ N $, as long as $ \log K =o(N)$. Thus, we should have $ \max_{m,k}  \|\wh\Sigma_{x\varepsilon}^{(k, m)} \| = O_p(\sqrt{(\log K) / N}) $.
Since $  \max_{k,m}   \|\wh\Sigma_{xx}^{(k, m)} \|=O_p(1) $ and $ \| \wh\Sigma_{xx}^{-1}  \wh\Sigma_{x\varepsilon}  \| = O_p( 1/\sqrt{N}) $, we can conclude that $  \max_{1\le k\le K} \big\|\wh B^{(k)} \widehat{\theta}_\text{ols}^* -\widehat{\Sigma}_{xy}^{(k)*} \| =O_p(\sqrt{(\log K) / N})$ as long as $ \log K = o(N) $.

The results above demonstrate that $ \max_{1\le k \le K} \| \wh E^{(k)}\| =O_p(\sqrt{(\log K) / N}) $ provided $ \log K =o(N) $.
These results further suggest that differences between the SMGDM estimators and the global OLS estimator are uniformly bounded by $ o_p( 1/\sqrt{N}) $, as long as $ \alpha \sqrt{\log K} \to 0$ and $ \log K = o(N) $.
Consequently, they should share the same statistical efficiency as the global OLS estimator.

\csection{NUMERICAL STUDIES}

In this section, we conduct extensive numerical experiments to corroborate our theoretical findings. First, we study the numerical convergence rate of the FMGDM method. Second, we explore the statistical properties of the stable solutions corresponding to the different tuning parameters specifications. Last, we compare the performances of the different types of MGDM methods.

\csubsection{Numerical Convergence}

In this subsection, we demonstrate how the tuning parameters  $ \alpha $ and $ \gamma $ affect the numerical convergence speed of the FMGDM method.
Specifically, we consider a standard linear regression model \eqref{linear model} with $ N = 5000 $ and $ p=50 $. The random noise $ \varepsilon_i $'s are simulated from the standard normal distribution with mean $ 0 $ and variance $ \sigma^2 =1 $. The true regression coefficient $ \theta_0 $ is given by $ \theta_0= (\theta_{01},\dots, \theta_{0p})^\top $ with $ \theta_{0j}=10\exp(-0.5 j) $ for $ 1\le j \le p $. Furthermore, the covariate $ X_i $'s are simulated from a multivariate normal distribution $ N_p(\0, \Sigma_{xx}) $, where $ \Sigma_{xx} = I_p + \kappa\1_p \1_p^\top$. One can easily verify that $\lambda_1 =  \lambda_{\max}(\Sigma_{xx})=\kappa p+1 $ and $\lambda_p =  \lambda_{\min}(\Sigma_{xx}) = 1 $. Hence the condition number of $ \Sigma_{xx} $ is given by $ \kappa p+1 $. To compare the convergence speed of the fixed minibatch-based GD (FMGD, i.e., $ \gamma = 0 $) and FMGDM methods, we consider two different specifications: (i) (FMGD) $ \alpha = 2/(\lambda_{1} + \lambda_{p}) $ and $ \gamma  = 0 $ and (ii) (FMGDM) $ \alpha = 4/(\sqrt{\lambda_{1} }+ \sqrt{\lambda_p})^2 $ and $ \gamma = (\sqrt{\lambda_{1}} - \sqrt{\lambda_p})^2 / (\sqrt{\lambda_{1}} + \sqrt{\lambda_p})^2 $. As indicated by Theorem \ref{thm:convergence_speed}, these two specifications correspond to the optimal tuning parameters configurations of the FMGD and FMGDM methods, respectively. Different $ \kappa $ values are studied (i.e., $ \kappa = 1, 2, 3 $), and the minibatch size is fixed to $ n=500 $. Accordingly, the number of minibatches is given by $ M = N/n=10 $.
For each setting, we first compute the corresponding stable solution on the $ M $-th minibatch $ \wh \theta^{(M)}$ by solving the linear dynamical system \eqref{dynamical system} directly. Next, we study how the algorithm converges to it.
To this end, we run the algorithm for a total of $ T=30 $ epochs.
Recall that $ \widehat{\theta}^{(t, M)} $ is the estimate obtained from the $ M $-th minibatch in the $ t $-th epoch.
We then record the difference $\delta^{(t)} = \| \wh \theta^{(t, M)} -  \wh \theta^{(M)}\|$ for each epoch $ 1 \le t\le T $.
We randomly replicate the experiments a total number of $ 100 $ times.
This leads to a total of $ 100 $ $ \delta^{(t)}  $ values for each epoch $ 1\le t\le T $.
The medians of these values are then plotted in Figure \ref{fig:numerical_convergences}.

\begin{figure}[htbp]
	\centering
	\includegraphics[width = \textwidth]{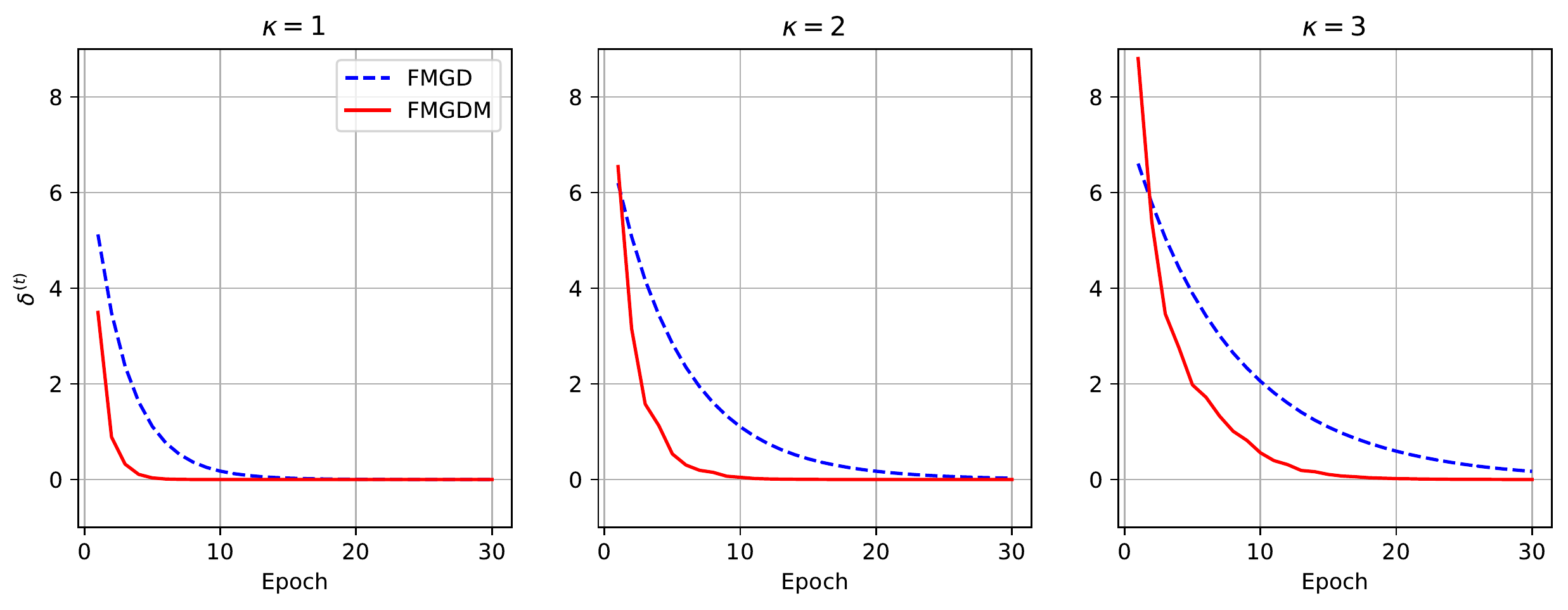}
	
	\caption{The numerical convergence rate of different algorithms (i.e., the FMGD and FMGDM methods). The vertical axis corresponds to the $ \delta^{(t)} $ value, while the horizontal axis corresponds to the epoch. The red solid line stands for the FMGDM method and the blue dashed line stands for the FMGD method.  }
	\label{fig:numerical_convergences}
\end{figure}

From Figure \ref{fig:numerical_convergences}, we can obtain the following interesting observations. First, we find that the curve of the FMGDM method converges towards $ 0 $ much faster than that of the FMGD method. This result confirms that our theoretical claims made in Theorem \ref{thm:convergence_speed}. Second, we can see that the convergence rates of both methods become slower as the conditional number increases. Moreover, the FMGDM method exhibits a more substantial advantage over the FMGD method. These findings also corroborate our theoretical results in Theorem \ref{thm:convergence_speed} (a) and (b) very well.

\csubsection{Statistical Efficiency of the Stable Solution}

This subsection is devoted to verifying the theoretical claims made in Section 2.5, that is, how the two tuning parameters (i.e., $ \alpha $ and $ \gamma $) affect the statistical efficiency of the stable solution. Specifically, we consider two different simulation experiments. One experiment studies the effects of the learning rate $ \alpha $ with a fixed momentum parameter $ \gamma $. The other experiment examines the effect of the momentum parameter $ \gamma $ while holding the learning rate $ \alpha $ fixed. In the first case, we fix $ \gamma =0.5 $ and let $ \alpha=0.15,0.1,0.05,0.01,0.001$, while in the second case, we fix $ \alpha=0.1 $ and let $ \gamma = 0, 0.5, 0.8, 1, 2,5, 10$.
Subsequently, we generate the simulated dataset as in Section 3.1 with $ \kappa = 1 $.
For each given $ (\alpha, \gamma) $-specification, the stable solution is computed for the last minibatch as $ \wh \theta^{(M)} $ by solving the linear dynamical system \eqref{dynamical system} directly. Then we evaluate the efficiency of different estimators by computing the estimation error (EE) as $ \|  \wh \theta^{(M)} - \theta_0\|^2 $, where $ \theta_0 $ is the true regression coefficient. For comparison, we also compute the EE values of the OLS estimator. This experiment is randomly replicated a total of $ 100 $ times, leading to a total of $ 100 $ EE values for each estimator, which are then log-transformed and shown as boxplots in Figure \ref{fig:stable_solution}.

\begin{figure}[htbp]
	\centering
	\includegraphics[width = \textwidth]{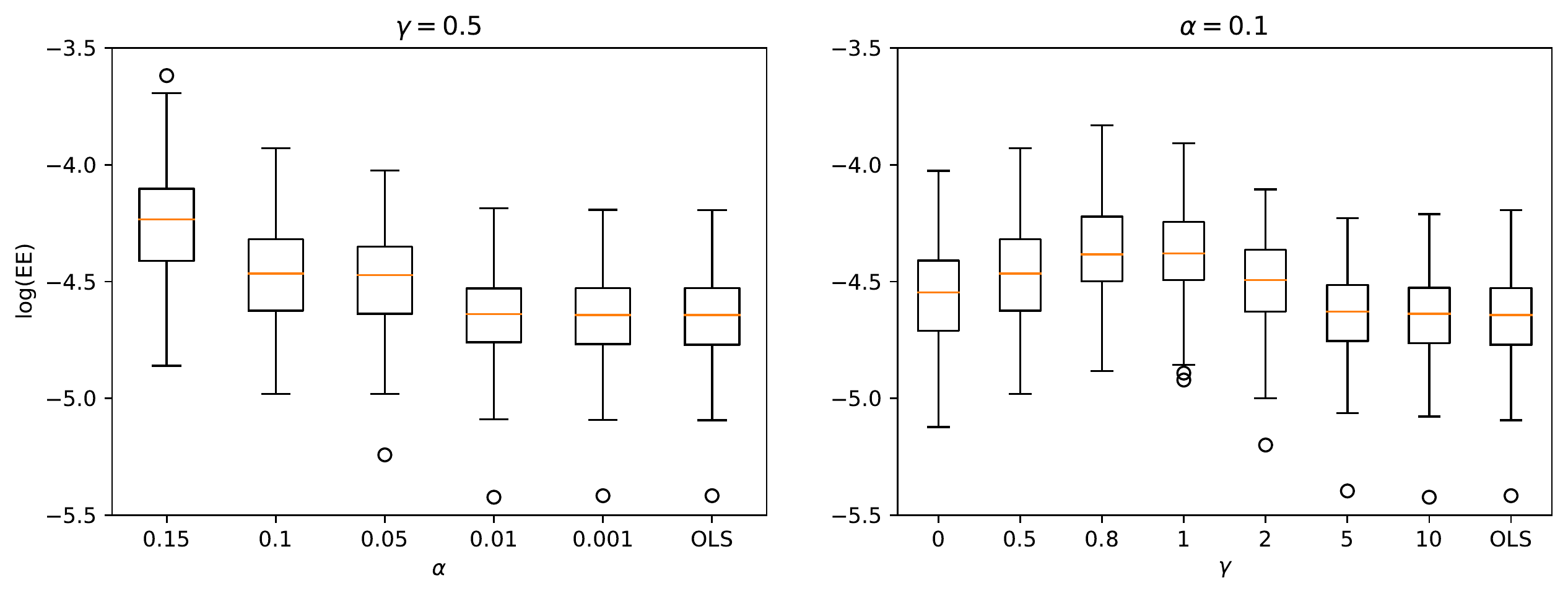}
	
	\caption{Boxplots of $ \log(\operatorname{EE}) $ of the FMGDM method under different $(\alpha, \gamma) $-specifications. In the left panel, we fix the momentum parameter as $ \gamma = 0.5 $ and let learning rate $ \alpha $ range from $ 0.15 $ to $ 0.001 $. In the right panel, we fix the learning rate as $ \alpha =0.1$ and let the momentum parameter $ \gamma $ range from $ 0 $ to $ 10 $.}
	\label{fig:stable_solution}
\end{figure}

From Figure \ref{fig:stable_solution}, we can draw the following conclusions.
First, from the left panel, we can see that, when the learning rate is large (e.g., $ 0.2, 0.1 $), the EE values of corresponding stable solutions are much larger than that of the global OLS estimator.
As the learning rate $ \alpha $ decays, the EE values steadily drop and then converge to that of the OLS estimator.
These findings verify the results in Theorem \ref{thm:ss} (a) and Theorem \ref{thm:asy_normality}. That is, by letting $ \alpha\to 0 $, the stable solution should have the same statistical efficiency as the global OLS estimator.
Second, from the right panel of Figure \ref{fig:stable_solution} we can see that when $ \gamma$ is small (e.g., $ \gamma\le 2 $), the EE values of the corresponding stable solutions are much larger than that of the OLS estimator. As $ \gamma $ increases, the EE values finally converge to that of the OLS estimator.
This finding is in agreement with the claim in Theorem \ref{thm:ss} (b).
Finally, we find that the EE values of the stable solution seem to increase from $ \gamma=0 $ to $ \gamma = 0.8 $.
To investigate this phenomenon, recall that by Theorem \ref{thm:asy_normality}, the extra variance term $ \alpha^2 M d_\gamma^2 \Sigma_{xx}  $ should be the dominant term of the EE value. Here, $ d_\gamma $ is the only factor involving $ \gamma $.
Therefore, we calculate values of $ d_\gamma $ for $ \gamma $ ranging from $ 0 $ to $ 10 $ by the explicit formulas in Appendix A.6. These values are then plotted, as shown in Figure \ref{fig:d_gamma}. It can be easily seen that the curve of $ d_\gamma $ has a similar pattern as that of the EE values in the right panel of Figure \ref{fig:stable_solution}. That is, the value of $ d_\gamma $ first increases with increasing $ \gamma $ from $ 0 $ to some value slightly smaller than $ 1 $, and then, $ d_\gamma $ decreases as $ \gamma $ increases.
This result further validates the asymptotic variance formula given in Theorem \ref{thm:asy_normality}.

\begin{figure}[htbp]
	\centering
	\includegraphics[width = 0.75\textwidth]{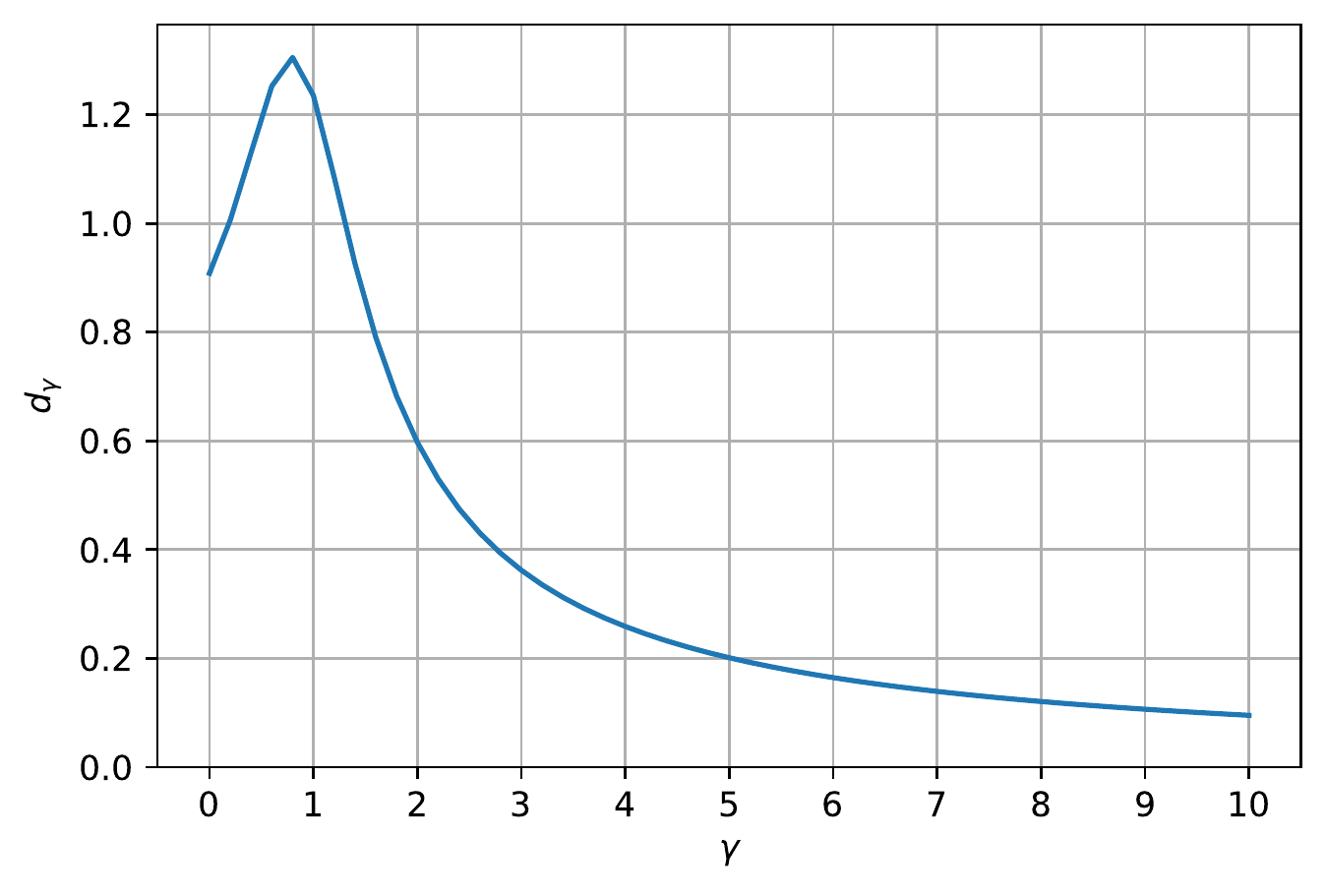}
	
	\caption{The curve of $ d_\gamma $ with $ \gamma $ ranging from $ 0 $ to $ 10 $. The values of $ d_\gamma $ are calculated by the explicit formulas \eqref{eq:a_gamma>0} and \eqref{eq:a_gamma=0} in Appendix A.6.}
	\label{fig:d_gamma}
\end{figure}

\csubsection{Different Types of MGDM Methods}

Next, we compare the performance of the different types of MGDM methods.
Specifically, we consider three different types of minibatch methods, i.e., fixed, shuffled and random minibatches. The corresponding MGDM methods are referred to as FMGDM, SMGDM and RMGDM methods, respectively.
For the RMGDM method, we use simple random sampling without replacement to generate each minibatch.
In this experiment, we consider three different learning rates as $ \alpha = 0.04, 0.02, 0.01 $ and fix the momentum parameter as $ \gamma = 0.9 $ by convention.
Subsequently, we simulate the dataset as in Section 3.1 with $ \kappa = 1 $.
Once the data are generated, we run each MGDM algorithm for a total of $ T=50 $ epochs and record the resulting estimator $ \wh \theta^{(T, M)} $ from the $ M $-th minibatch. For a fair comparison, passing over every $ M = N/n = 10$ minibatches is considered as one epoch for each MGDM method.
Then we compute the corresponding estimation error (EE) as $ \|\wh \theta^{(T, M)} -\theta_0\|^2  $ for each method under each $ (\alpha, \gamma) $-specification.
For comparison, we also compute the EE values of the OLS estimator.
We randomly replicate the experiment a total of $ 100 $ times, leading to a total of $ 100 $ EE values for each estimator.
These EE values are then log-transformed and shown as boxplots in Figure \ref{fig:MGDMs}.

\begin{figure}[htbp]
	\centering
	\includegraphics[width = 0.75\textwidth]{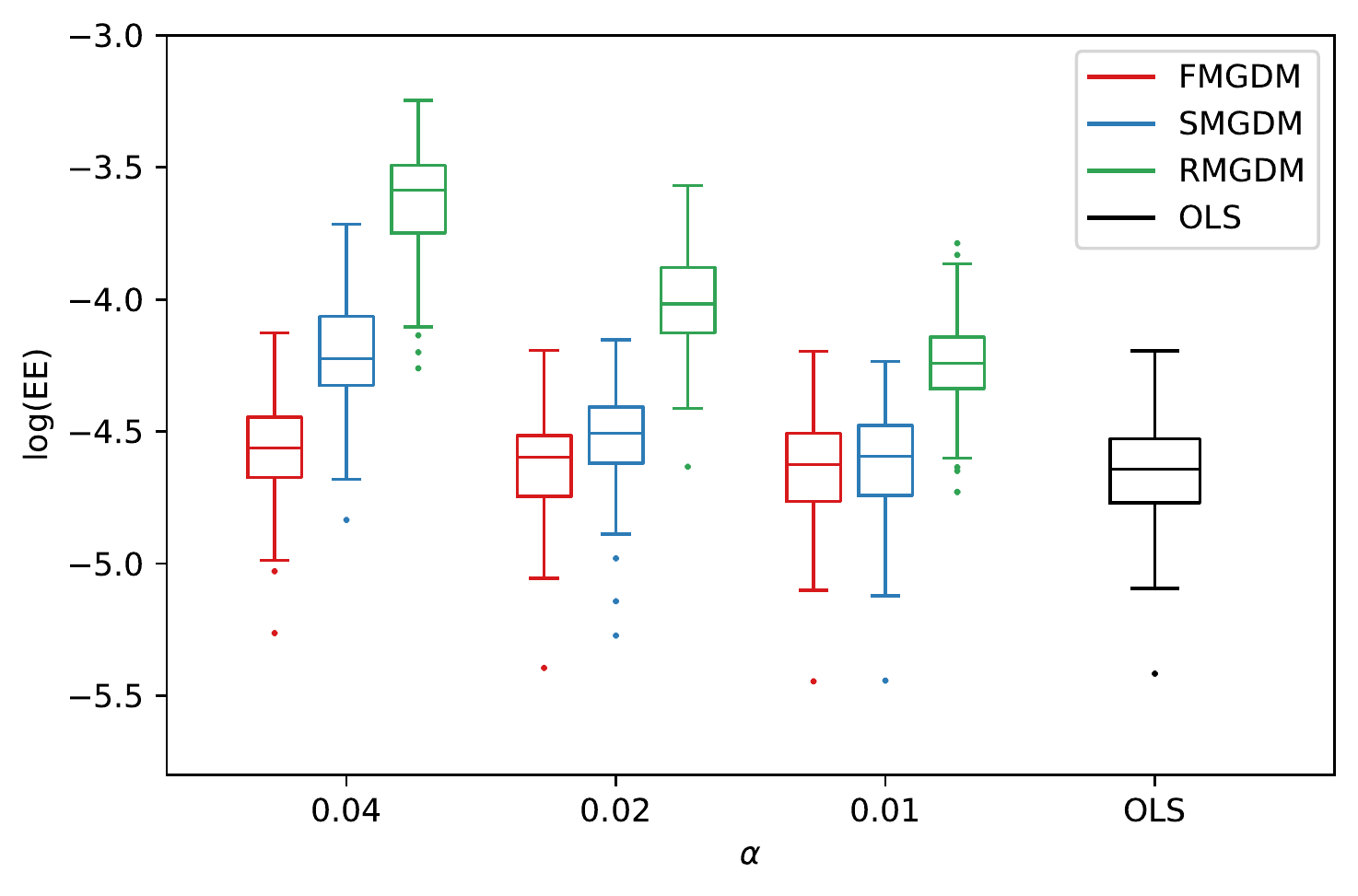}
	
	\caption{Boxplots of $ \log(\operatorname{EE}) $ of the FMGDM, SMGDM, RMGDM and OLS estimators. For the three MGDM methods, we consider three different learning rates  $ \alpha = 0.04, 0.02, 0.01 $ and fix the momentum parameter as $ \gamma = 0.9 $. Red boxes represent the FMGDM estimator, blue boxes represent the SMGDM estimator, green boxes represent the RMGDM estimator, and the black box represents the OLS estimator.}
	\label{fig:MGDMs}
\end{figure}

From Figure \ref{fig:MGDMs}, we immediately observe that the FMGDM estimator has the smallest EE values among the three MGDM methods.
For the SMGDM method, we can find that the corresponding EE values converge to that of the OLS estimator as the learning rate $ \alpha $ decreases.
This finding corroborates our justifications in Section 2.6 very well.
That is, the SMGDM estimator should enjoy the same statistical efficiency as the OLS estimator with a learning rate of $ \alpha \to 0 $.
The EE values of the RMGDM estimator also decrease with decreasing learning rate.
However, there is still a clear gap between the RMGDM estimator and the OLS estimator in terms of EE values, even though the learning rate is sufficiently small (i.e., $ \alpha = 0.01 $).
This result may be due to the extra random errors created by simple random sampling.
These results demonstrate that the FMGDM and SMGDM methods are more preferable in terms of statistical efficiency compared to the RMGDM method.

\csubsection{Application to Airline Data}

To demonstrate the performance of the MGDM methods on a real-world dataset, we consider here the U.S. Airline Dataset. The whole dataset is available for download at \url{http: //stat-computing.org/dataexpo/2009}. It contains the flight arrival and departure details for all commercial flights within the USA from 1987 to 2008. We use the data from 2008, and aim to predict the delay time in the arrival of a flight given other flight information.
Each record of the data contains the arrival delay, departure time, scheduled arrival and departure time, flight date,  distance of the flight, information of the carrier, origin and destination. The detailed variable information is described in Table \ref{table:airline}. 
The six continuous variables are standardized to be mean 0 and variance 1, and the five categorical variables are converted to dummy variables with appropriate dimensions. Finally, a total of 126 predictors and one response variable are used in the linear regression model, and the total sample size is over 2.5 million observations.

In this experiment, we set the minibatch size as $ n=1000  $, momentum parameter as $ \gamma = 0.9 $ and consider three different learning rates $ \alpha = 0.4, 0.2, 0.1 $.
For comparison purpose, we first randomly choose a subsample of size $ N= 10^5 $ from the whole sample, and then run each MGDM algorithm with $ M = N/n = 100 $ for a total of $T= 100 $ epochs. Next, we compute the estimation error  (EE) as $ \| \wh\theta^{(T, M)} -\wh \theta_\text{ols}\|^2$ for each method under each $ (\alpha, \gamma) $-specification, where $\wh \theta_\text{ols} $ is the OLS estimator of the regression coefficient based on the whole sample. We replicate the experiment through $ 50 $ times random subsampling, leading to a total of $ 50 $ EE values for each estimator. These EE values are then log-transformed and shown as boxplots in Figure \ref{fig:airline}. From Figure \ref{fig:airline}, we can see a similar pattern as shown in Figure \ref{fig:MGDMs}. That is, as the learning rate $ \alpha $ decreases, the EE values of all three methods decrease. In addition, the FMGDM and SMGDM methods show a competitive performance. The RMGDM method has the worst performance under all three $ \alpha $ specifications.

\begin{figure}[htbp]
	\centering
	\includegraphics[width = 0.75\textwidth]{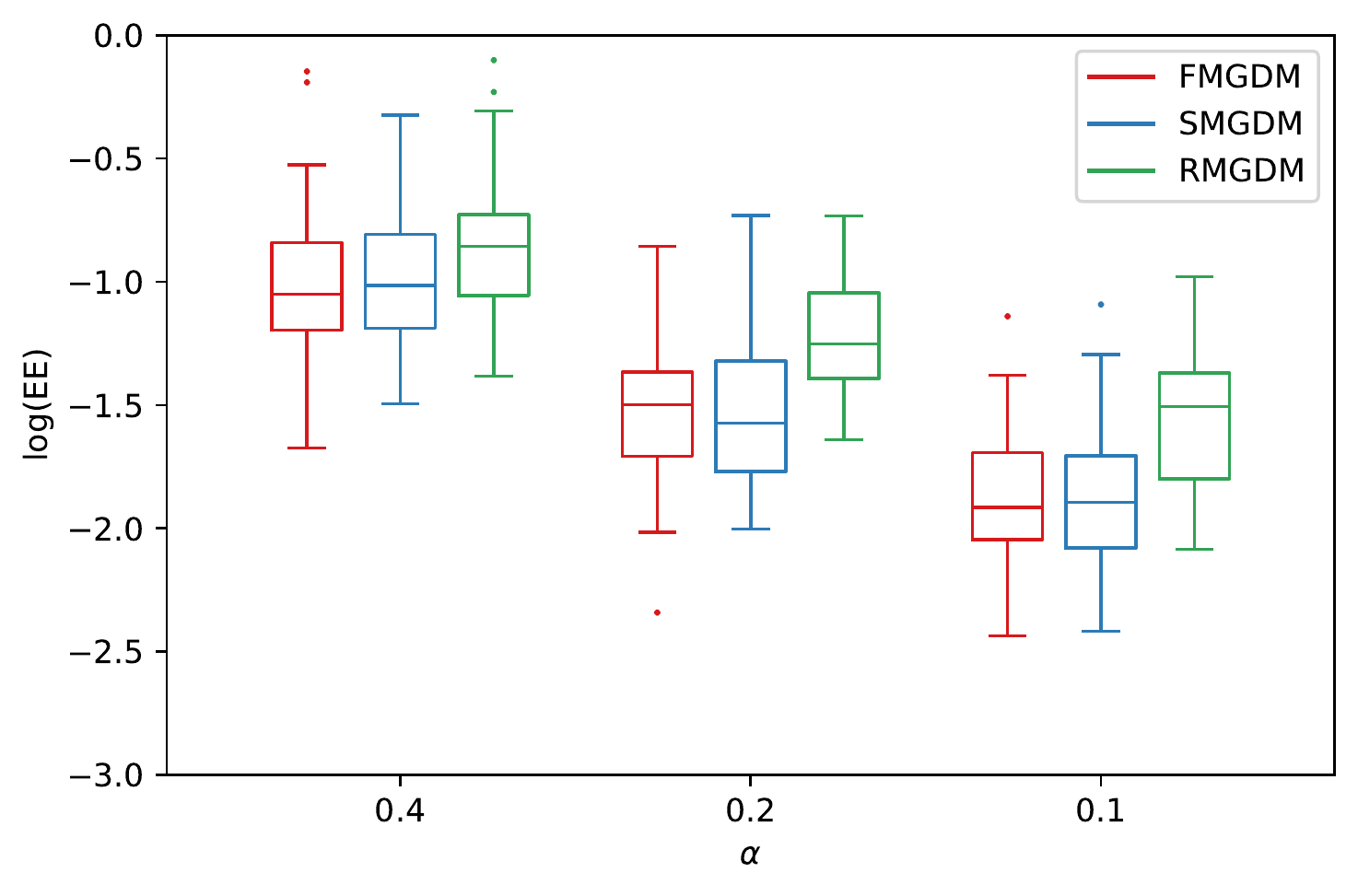}
	
	\caption{Boxplots of $ \log(\operatorname{EE}) $ of the FMGDM, SMGDM, and RMGDM estimators based on the airline data. For the three MGDM methods, we consider three different learning rates  $ \alpha = 0.4, 0.2, 0.1 $ and fix the momentum parameter as $ \gamma = 0.9 $. Red boxes represent the FMGDM estimator, blue boxes represent the SMGDM estimator, green boxes represent the RMGDM estimator.}
	\label{fig:airline}
\end{figure}

\begin{sidewaystable}
\caption{Variable description for the U.S. airline data from 2008. Numerical variables are standardized to have mean 0 and variance 1.}
\label{table:airline}
\centering
\begin{tabular}{lll}
	\toprule
	Variable&Description&Variable used in the model\\
	\midrule
	\verb|ArrDelay|&How long was the delay in the arrival of the flight &Used as the numerical response variable\\
	\midrule
	\verb|Month|&Which month of the year&Converted to 11 dummies\\
	\verb|DayofMonth|&Which day of the month&Used as numerical variable\\
	\verb|DayOfWeek|&Which day of the week&Converted to 6 dummies\\
	\verb|DepTime|&Actual departure time&Used as numerical variable\\
	\verb|CRSDepTime|&Scheduled departure time&Used as numerical variable\\
	\verb|CRSArrTime|&Scheduled arrival time&Used as numerical variable\\
	\verb|Distance|&Distance between the origin and destination in miles&Used as numerical variable\\
	\verb|UniqueCarrier|&Flight carrier code for 29 carriers&Top 7 carries converted to 6 dummies\\
	\verb|Origin|&Departing origin &Top 50 origin cities converted to 49 dummies\\
	\verb|Dest|&Destination of the flight&Top 50 origin cities converted to 49 dummies\\
	\toprule
\end{tabular}
\end{sidewaystable}

\csection{CONCLUDING REMARKS}

In this paper, we study the theoretical properties of MGDM algorithms for linear regression models.
We focus mainly on two different types of MGDM methods, i.e., fixed MGDM (FMGDM) and shuffled MGDM (SMGDM) methods.
We first investigate the FMGDM method by introducing a linear dynamical system.
Then, we provide nearly sufficient and necessary conditions for numerical convergence.
Furthermore, the numerical convergence rate of FMGDM algorithm is investigated.
This analysis leads to the theoretically optimal specification of the tuning parameters.
Theoretical findings reveal that the extra momentum term could greatly speed up numerical convergence compared to a standard gradient descent method without momentum.
This is particularly true for problems with large condition numbers.
Subsequently, the statistical properties of the resulting FMGDM estimator are carefully investigated and the asymptotic normality is further established.
We then find that with a sufficiently small learning rate, the MGDM estimators (both fixed and shuffled) can share the same statistical efficiency as the global estimator.
Finally, these theoretical results are fully validated through extensive numerical studies.
To conclude this article, we discuss here several interesting topics for future study.
First, we examine the FMGDM algorithm for the quadratic objective function, due to its analytical simplicity.
Of note, many learning problems have non-quadratic objection functions, such as logistic regression and quantile regression models. However, it seems to us that the theoretical properties of either the fixed or shuffled MGDM estimators for these problems remain unknown. This is a research topic of great interest.
Second, modern statistical inference often involves ultrahigh-dimensional problems.
It is valuable to extend our results on the MGDM methods to the models with diverging $ p $.
Last, there are many momentum-related methods (e.g., Nesterov's momentum method,  RMSprop and ADAM) that are widely used in practice.
Therefore, it is also very important to investigate their theoretical properties under a fixed or shuffled minibatch setup.

\renewcommand{\theequation}{A.\arabic{equation}}
\setcounter{equation}{0}

\scsection{APPENDIX A}

\scsubsection{Appendix A.1 Proof of Theorem \ref{thm:invertible}}

Construct a block matrix analogous to $ \wh\Omega $ in \eqref{Omega_hat} as
\begin{align} \label{Omega}
	\Omega = \begin{bmatrix}
		I_p&\0&\0&\cdots&\0&\gamma I_p&-\Delta\\
		-\Delta&I_p&\0&\cdots&\0&\0&\gamma I_p\\
		\gamma I_p&-\Delta&I_p&\0&\cdots&\0&\0\\
		\vdots&\vdots&\vdots&\vdots&\vdots &\vdots&\vdots\\
		\0&\0&\0&\cdots&\gamma I_p&-\Delta&I_p
	\end{bmatrix}\in \mR^{q\times q},
\end{align}
where $ \Delta = (1+\gamma) I_p - \alpha \Sigma_{xx} \in \mR^{p \times p}$. By Lemma \ref{lemma:invertible} in the Appendix B, we know that $ \Omega $ is invertible. Note that
\begin{align*}
	\wh\Omega - \Omega  = \alpha \begin{bmatrix}
		\0&\0&\cdots&\0&\0&\widehat{\Sigma}_{xx}^{(1)} - \Sigma_{xx}\\
		\widehat{\Sigma}_{xx}^{(2)} - \Sigma_{xx}&\0&\cdots&\0&\0&\0\\
		\0&\widehat{\Sigma}_{xx}^{(3)} - \Sigma_{xx}&\cdots&\0&\0&\0\\
		\vdots&\vdots&\vdots&\vdots&\vdots&\vdots\\
		\0&\0&\cdots&\0&\widehat{\Sigma}_{xx}^{(M)} - \Sigma_{xx}&\0\\
	\end{bmatrix}.
\end{align*}
By law of large number, we should have $ \|\widehat{\Sigma}_{xx}^{(m)} - \Sigma_{xx} \|\to_p 0 $ as $ n \to \infty$ for each $ 1\le m\le M $.
Hence, $ \| \wh\Omega - \Omega \|\to_p 0 $ as $ n \to \infty $.
Since determinant is a continuous function, we should have $ |\wh\Omega| \to_p |\Omega|\ne 0 $ as $ n \to \infty $. Thus, the claim of the theorem follows immediately.

\scsubsection{Appendix A.2 Proof of Theorem \ref{thm:convergence}}

By \eqref{FMBM1} and \eqref{FMBM2}, we have
\begin{align*}
	\widehat{\theta}^{(t, m)} &= 	\widehat{\theta}^{(t, m-1)} - \gamma(	\widehat{\theta}^{(t, m-2)}-	\widehat{\theta}^{(t, m)} ) -\alpha \Big\{ \widehat{\Sigma}_{xx}^{(m)}	\widehat{\theta}^{(t, m-1)}  -  \widehat{\Sigma}_{xy}^{(m)}\Big\}\\
	&= \Big\{(1+\gamma)I_p - \alpha \widehat{\Sigma}_{xx}^{(m)} \Big\} 	\widehat{\theta}^{(t, m-1)} - \gamma 	\widehat{\theta}^{(t, m-2)} +\alpha \widehat{\Sigma}_{xy}^{(m)}.
\end{align*}
Note the stable solution also satisfies $ \widehat{\theta}^{(m)} = \Big\{(1+\gamma)I_p - \alpha \widehat{\Sigma}_{xx}^{(m)} \Big\} 	\widehat{\theta}^{(m-1)} - \gamma 	\widehat{\theta}^{(m-2)} +\alpha \widehat{\Sigma}_{xy}^{(m)}  $. Hence, taking a difference of above two equations gives
\begin{align*}
	\widehat{\theta}^{(t, m)}- \widehat{\theta}^{(m)} &= \Big\{(1+\gamma)I_p - \alpha \widehat{\Sigma}_{xx}^{(m)} \Big\} 	\Big(	\widehat{\theta}^{(t, m-1)} - \widehat{\theta}^{(m-1)}\Big) - \gamma \Big(	\widehat{\theta}^{(t, m-2)} -	\widehat{\theta}^{(m-2)} \Big),
\end{align*}
where $ \widehat{\theta}^{(t, 0)} =  \widehat{\theta}^{(t-1, M)} $ and $ \widehat{\theta}^{(t, -1)} =  \widehat{\theta}^{(t-1, M-1)} $. If we denote $ \delta^{(t, m)} =  	\widehat{\theta}^{(t, m)}- \widehat{\theta}^{(m)}$, then we have
\begin{align*}
	r^{(t, m)}=\begin{bmatrix}
		\delta^{(t, m)}\\
		\delta^{(t, m-1)}
	\end{bmatrix} = \begin{bmatrix}
		\wh\Delta^{(m)}& -\gamma I_p\\
		I_p& \0
	\end{bmatrix}\begin{bmatrix}
		\delta^{(t, m-1)}\\
		\delta^{(t, m-2)}
	\end{bmatrix}= D^{(m)} r^{(t,m-1)}.
\end{align*}
Consequently, we obtain the recursion formula as $ r^{(t,m)} = C^{(m)}  r^{(t-1,m)}$, where
\begin{align*}
	&C^{(1)}  = D^{(1)} D^{(M)}  D^{(M-1)} \cdots D^{(2)}, \\
	&C^{(m)}  = D^{(m)}  \cdots D^{(1)} D^{(M)} \cdots D^{(m+1)},\ 2\le m\le M-1,\\
	&C^{(M)}  = D^{(M)} D^{(M-1)} \cdots \cdots D^{(1)}.
\end{align*}

\noindent\textbf{Proof of (a)}. Consider the block matrix $ D = \big[\Delta, -\gamma I_p; I_p, \0 \big] \in \mR^{(2p) \times (2p)} $, where $ \Delta = (1+\gamma) I_p - \alpha \Sigma_{xx} \in \mR^{p \times p}$. By assumed conditions in (a) and Lemma \ref{lemma:SR}(a), we know that $ \rho(D)<1 $. Note that $ \|D^{(m)} - D\| = \alpha \|\widehat \Sigma_{xx}^{(m)} -\Sigma_{xx} \| \to_p 0  $ as $ n \to \infty $. Since $ M $ is fixed, it follows immediately that $ \|C^{(m)} - D^M\|\to_p 0  $ as $ n \to \infty $. Because the spectral radius is a continuous function of the entries of a matrix, we obtain that $ \rho(C^{(m)}) \to_p \rho(D^M) = \rho^M(D) <1$  as $ n\to \infty $.  Consequently, we should have $ \rho\{ (C^{(m)} )^t \}= \rho^t(C^{(m)} ) \to_p 0$  as $ n\to \infty $. Then by recursion formula, we have $ r^{(t,m)} = \{C^{(m)}\}^t  r^{(0,m)} \to_p \0$ as $ n\to \infty $. Thus, (a) follows immediately.

\noindent\textbf{Proof of (b)}. By assumed conditions in (b) and Lemma \ref{lemma:SR}(e), we have $ \rho(D)>1$. Then using the similar procedures in the proof of (a), we can obtain that $ \rho(C^{(m)}) \to_p \rho^M(D) > 1$ as $ n\to \infty $. Consequently, $ \{C^{(m)}\}^t \not\to \0$ as $ t\to \infty $ holds with probability tending to one as $ n \to \infty$ \citep[Theorem 5.6.12]{horn2012matrix}. Thus, as long as $\delta^{(0, m)}\ne \0  $, $ \delta^{(t, m)} $ will not converge to zero with probability tending to one as $ n \to \infty$. We have proved (b) and complete the proof.

\scsubsection{Appendix A.3 Proof of Theorem \ref{thm:convergence_speed}}

From the proof of Theorem \ref{thm:convergence}, we know that $ \gamma\le \rho(D) <1$ if $ 0\le \gamma <1 $ and $0<\alpha< 2(1+\gamma)/\lambda_{1}$, and $ \rho(C^{(m)}) = \rho^M(D) + \zeta_n $, where $ \zeta_n\to_p 0 $ as $ n\to  \infty $.
Gelfand formula \citep[Corollary 5.6.14]{horn2012matrix} guarantees that $ \lim_{t \to \infty} \| \{C^{(m)} \}^t \|^{1/t} = \rho(C^{(m)})$. Then we have $ \| \{C^{(m)} \}^t \| = (\rho(C^{(m)}) + \eta_t)^t $, where $ \eta_t\to 0 $ as $ t \to \infty $. Recall that $ r^{(t,m)} = \{C^{(m)}\}^t  r^{(0,m)}$. Consequently, $ \|\delta^{(t, m)}\| \le \|r^{(t,m)}\| \le \|\{C^{(m)}\}^t\|\ \|r^{(0,m)}\|   \le  \{\rho^M(D) + \zeta_n + \eta_t\}^t \big(\|\delta^{(0, m)} \| + \|\delta^{(0, m-1)} \| \big)$. Then the conclusions (a), (b) and (c) follow directly from Lemma \ref{lemma:SR}(c), (b) and (d) in Appendix B, respectively.

\scsubsection{Appendix A.4 Proof of Theorem \ref{thm:ss}}

Recall that $ P = [P_1, P_2] $ is a orthogonal matrix, where $ P_1 = I^* / \sqrt{M} \in\mR^{q\times p}$. In fact, we can let $ P_2 = Z\otimes I_p \in \mR^{q \times(q-p)}$, where
\begin{align}\label{Z_mat}
	Z=\begin{bmatrix}
		1 / \sqrt{2} & 1 / \sqrt{6} & \cdots & 1 / \sqrt{M(M-1)} \\
		-1 / \sqrt{2} & 1 / \sqrt{6} & \cdots & 1 / \sqrt{M(M-1)} \\
		0, & -2 / \sqrt{6} & \cdots & 1 / \sqrt{M(M-1)} \\
		\vdots & \vdots & \vdots & \vdots \\
		0 & 0 & \cdots & -(M-1) / \sqrt{M(M-1)}
	\end{bmatrix}\in \mR^{M \times (M-1)}.
\end{align}

\noindent\textbf{Proof of (a)}. We will prove (a) in two steps. In the first step, we obtain the expression of $ \wh\Omega^*= (P^\top \wh\Omega P)^{-1} $. In the second step, we investigate stable solution by $ \widehat{\theta}^* = \alpha P(P^\top \wh\Omega P)^{-1}P^\top \widehat{\Sigma}_{xy}^*  $.

\textsc{Step 1.}
For the simplicity, we further denote $ \wh\Omega^*= [\wh\Omega_{11}^*, \wh\Omega_{12}^*; \wh\Omega_{21}^*, \wh\Omega_{22}^*] $, where $ \wh\Omega_{11}^*\in\mR^{p\times p} $, $ \wh\Omega_{12}^*\in\mR^{p\times (q-p)} $, $ \wh\Omega_{21}^*\in\mR^{(q-p)\times p} $, and $ \wh\Omega_{22}^*\in\mR^{(q-p)\times (q-p)} $. We then compute each block of $  \wh\Omega^* $ subsequently.
By techniques used in the proof of Lemma \ref{lemma:invertible}, one can verify that $ A $ is diagonalizable with $ q-p $ non-zero eigenvalues. Hence, $\rank(A) = \rank(P^\top A P) = \rank(A_{22}^p) = q-p$. i.e., $ A_{22}^p $ is invertible.  Note that $ \wh B_{11}^P =   P_1 ^\top \wh B P_1  = M^{-1}\sum_{m=1}^M \widehat{\Sigma}_{xx}^{(m)}=\widehat{\Sigma}_{xx}$ is also invertible almost surely. Consequently, it follows from the formula for the inverse of a block matrix \citep{seber2008matrix} and Taylor expansion that
\begin{align*}
	\wh\Omega_{11}^* &=  \Big[\alpha \wh B_{11}^P - (\alpha \wh B_{12}^P) (A_{22}^P+\alpha \wh B_{22}^P)^{-1} (\alpha \wh B_{21}^P) \Big]^{-1}\nonumber\\
	&=\alpha^{-1} \Big[ \widehat{\Sigma}_{xx}- \alpha \wh B_{12}^P\Big\{ (A_{22}^P)^{-1}  -(A_{22}^P)^{-1}  \wh B_{22}^P (A_{22}^P)^{-1} +O_p(\alpha^2) \Big\} \wh B_{21}^P \Big]^{-1}\\
	&=\alpha^{-1} \widehat{\Sigma}_{xx}^{-1} +\widehat{\Sigma}_{xx}^{-1}    \wh B_{12}^P (A_{22}^P)^{-1} \wh B_{21}^P   \widehat{\Sigma}_{xx}^{-1}  - \alpha \widehat{\Sigma}_{xx}^{-1}\wh B_{12}^P(A_{22}^P)^{-1}  \wh B_{22}^P (A_{22}^P)^{-1}\wh B_{21}^P \widehat{\Sigma}_{xx}^{-1} +O_p(\alpha^2).
\end{align*}
Similarly, it can be derived that
\begin{align*}
	\wh\Omega_{12}^* =&  -\wh\Omega_{11}^*  (\alpha \wh B_{12}^P)  (A_{22}^P+\alpha \wh B_{22}^P)^{-1} \\
	=& - \widehat{\Sigma}_{xx}^{-1} \wh B_{12}^P(A_{22}^P)^{-1} + \alpha \widehat{\Sigma}_{xx}^{-1} \wh B_{12}^P (A_{22}^P)^{-1} \wh B_{22}^P (A_{22}^P)^{-1} \\
	&- \alpha   \widehat{\Sigma}_{xx}^{-1}    \wh B_{12}^P (A_{22}^P)^{-1} \wh B_{21}^P   \widehat{\Sigma}_{xx}^{-1} \wh B_{12}^P(A_{22}^P)^{-1} +O_p(\alpha^2),\\
	\wh\Omega_{21}^* =& - (A_{22}^P+\alpha \wh B_{22}^P)^{-1} (\alpha \wh B_{21}^P) \wh\Omega_{11}^* \\
	=&-(A_{22}^P)^{-1} \wh B_{21}^P \widehat{\Sigma}_{xx}^{-1} + \alpha (A_{22}^P)^{-1} \wh B_{22}^P (A_{22}^P)^{-1}  \wh B_{21}^P \widehat{\Sigma}_{xx}^{-1}  \\
	&- \alpha (A_{22}^P)^{-1}\wh B_{21}^P \widehat{\Sigma}_{xx}^{-1}    \wh B_{12}^P (A_{22}^P)^{-1} \wh B_{21}^P   \widehat{\Sigma}_{xx}^{-1} +O_p(\alpha^2), \\
	\wh\Omega_{22}^* =& (A_{22}^P+\alpha \wh B_{22}^P)^{-1} + (A_{22}^P+\alpha \wh B_{22}^P)^{-1}  (\alpha \wh B_{21}^P)  \wh\Omega_{11}^*  (\alpha \wh B_{12}^P)  (A_{22}^P+\alpha \wh B_{22}^P)^{-1} \\
	=& (A_{22}^P)^{-1} - \alpha (A_{22}^P)^{-1} \wh B_{22}^P(A_{22}^P)^{-1}+ \alpha (A_{22}^P)^{-1}  \wh B_{21}^P \widehat{\Sigma}_{xx}^{-1} \wh B_{12}^P (A_{22}^P)^{-1}  +O_p(\alpha^2).
\end{align*}
Denote $ \wh H_1 = P_1 \widehat{\Sigma}_{xx}^{-1} P_1^\top $ and $ H_2 = P_2 (A_{22}^P)^{-1} P_2^\top $.
Then we should have
\begin{align}
 \wh\Omega^{-1}  =& P \wh\Omega^* P^\top = P_1 \wh\Omega_{11}^* P_1^\top + P_1 \wh\Omega_{12}^* P_2^\top + P_2 \wh\Omega_{21}^* P_1^\top + P_2 \wh\Omega_{22}^* P_2^\top \nonumber \\
 =& \alpha^{-1} \wh H_1 + (\wh H_1 \wh B H_2 \wh B \wh H_1 - \wh H_1 \wh B H_2 - H_2 \wh B \wh H_1 + H_2) \nonumber\\
 & - \alpha (\wh H_1 \wh B H_2 \wh B H_2 \wh B \wh H_1 -\wh H_1 \wh B H_2 \wh B H_2-H_2 \wh B H_2 \wh B \wh H_1+ H_2 \wh B H_2 ) \nonumber\\
 & + \alpha (  H_2 \wh B \wh H_1 \wh B H_2- \wh H_1 \wh B H_2 \wh B \wh H_1 \wh B H_2 - H_2 \wh B \wh H_1 \wh B H_2 \wh B \wh H_1) + O_p(\alpha^2) \nonumber\\
 =&\alpha^{-1} \wh H_1  + (\wh H_1 \wh B - I_q) H_2 (\wh B \wh H_1 - I_q)  - \alpha(\wh H_1 \wh B - I_q) (H_2 BH_2) (\wh B \wh H_1 - I_q)  \nonumber\\
 &+ \alpha (\wh H_1 \wh B - I_q) (H_2 \wh B \wh H_1 \wh B H_2)(\wh B \wh H_1 - I_q) - \alpha \wh H_1 \wh B H_2 \wh B \wh H_1 \wh B H_2 \wh B \wh H_1 +O_p(\alpha^2) \label{Omega_inv}.
\end{align}

\textsc{Step 2.}
Recall that $ P_1 = I^* / \sqrt{M} $. Hence, $ P_1^\top \widehat{\Sigma}_{xy}^* = M^{-1/2}\sum_{m=1}^M \widehat{\Sigma}_{xy}^{(m)}=\sqrt{M} \widehat{\Sigma}_{xy}  $.
Thus, $ \wh H_1   \widehat{\Sigma}_{xy}^* = \sqrt{M} P_1 \widehat{\Sigma}_{xx}^{-1} \widehat{\Sigma}_{xy} = I^* \wh \theta_\text{ols} = \widehat{\theta}_\text{ols}^*$.
Combine the results above, it can be concluded that
\begin{equation}\label{stable solution1}
	\widehat{\theta}^* = \alpha \wh\Omega^{-1}\widehat{\Sigma}_{xy}^* = \widehat{\theta}_\text{ols}^* +\alpha \widehat{E}\big\{ 1+O_p(\alpha) \big\},
\end{equation}
where
\begin{align*}
	\widehat{E} =& (\wh H_1 \wh B - I_q) H_2 (\wh B \wh H_1 - I_q)\widehat{\Sigma}_{xy}^*\\	
	=&\big(P_1 \widehat{\Sigma}_{xx} ^{-1}P_1^\top \wh B - I_q \big)  \big\{P_2 (A_{22}^P)^{-1} P_2^\top \big\} \big(\wh B \widehat{\theta}_\text{ols}^* -\widehat{\Sigma}_{xy}^*  \big). \label{E hat}
\end{align*}
We finish the proof of (a).

\noindent\textbf{Proof of (b)}. From the proof of (a), we can find that $ \gamma $ affects $ \wh\Omega^* $ only through $ A_{22}^P $. Note that $I_q = P P^\top  = P_1 P_1^\top + P_2 P_2^\top  $. Then we have
\begin{align*}
	\|A_{22}^P\|_F^2 = \tr\Big\{A_{22}^P (A_{22}^{P}) ^\top\Big\} &= \tr(P_2^\top A P_2 P_2^\top A^\top P_2) = \tr\Big\{A(I_q - P_1 P_1^\top) A^\top  (I_q - P_1 P_1^\top) \Big\}\\
	&=\tr(AA^\top) - 2\tr(A P_1 P_1^\top A^\top) + \tr(A P_1 P_1^\top A P_1 P_1^\top) \\
	&=\tr(AA^\top) = \big\{(1+\gamma)^2+\gamma^2+1\big\}q,
\end{align*}
where we have used the fact $  AP_1 = \0  $. Hence, we conclude that $ \|(A_{22}^P)^{-1}\| \to 0 $ as $ \gamma \to \infty $. Then by analogous procedures to derive \eqref{stable solution1}, we can obtain that $\widehat{\theta}^* = \widehat{\theta}_\text{ols}^* +\alpha \widehat{E} \Big\{1+O_p( \|(A_{22}^P)^{-1}\|  )\Big\}$,
where $ \wh E $ is the same as that in \eqref{stable solution1}.
The only difference is that the remainder term is $ O_p( \|(A_{22}^P)^{-1}\| ) $ instead of $ O_p( \alpha )$.
We have proved (b) and complete the proof.

\scsubsection{Appendix A.5 Proof of Theorem \ref{thm:asy_normality}}

Note that $ \widehat{\Sigma}_{xx}^{(m)} = n^{-1} \sum_{i \in \mS_{(m)} }X_i X_i^\top$ and $ \widehat{\Sigma}_{xy}^{(m)} = n^{-1} \sum_{i \in \mS_{(m)} }X_i Y_i$ are sample mean of $ n $ independent random variables. Then by central limit theorem, we should have $ \sqrt{n} (\widehat{\Sigma}_{xy}^{(m)}  -\widehat{\Sigma}_{xx}^{(m)}\theta_0 ) \to_d N(\0, \sigma^2\Sigma_{xx}) $ as $ N\to \infty $ for each $ 1\le m\le M $. Write the block diagonal matrix as $\wh \Sigma_{xx}^* = \diag\{ \wh \Sigma_{xx}^{(1)}, \dots, \wh \Sigma_{xx}^{(M)} \} \in \mR^{q\times q} $. Then it follows that $ \sqrt{n} \big\{ \wh\Sigma_{xy}^* - \wh\Sigma_{xx}^*(\1_M \otimes \theta_0)\big\}\to_d N_q(\0,\sigma^2 I_M \otimes \Sigma_{xx} ) $, where $ \widehat \Sigma_{xy}^* =   \big(\widehat\Sigma_{xy}^{(1)\top}, \dots, \widehat \Sigma_{xy}^{(M)\top}\big)^\top \in \mR^q$.
Recall that the stable solution has the form $ \widehat \theta^* = \alpha\wh \Omega ^{-1} \widehat \Sigma_{xy}^* $. Consequently, $\sqrt{n} \big\{\widehat \theta^* - \wh \mu(\alpha) \big\} \to_d N_q\big(\0, \sigma^2 V(\alpha) \big)$, where $ \wh \mu(\alpha) =  \alpha\wh \Omega ^{-1} \wh\Sigma_{xx}^*(\1_M \otimes \theta_0)$, and $ V(\alpha) = \alpha^2 \Omega^{-1} (I_M \otimes \Sigma_{xx}) \Omega^{-1} $. Here, $ \Omega $ is defined in \eqref{Omega}, and it is the probabilistic limit of $\wh \Omega $, i.e., $ \wh \Omega \to_p \Omega $ as $ N\to \infty $. In the following steps, we investigate $ \wh \mu(\alpha) $ and $ V(\alpha) $ subsequently.

\textsc{Step 1.} By the asymptotic form of $ \wh \Omega^{-1} $ in \eqref{Omega_inv}, we can show that $ \wh \mu(\alpha) = \1_M\otimes \theta_0 - \alpha^2 \wh H_1 \wh B H_2 \wh B \wh H_1 \wh B H_2 \wh \Sigma_{xx}^* (\1_M\otimes \theta_0) +O_p(\alpha^2)$. This is due to the fact that $ \wh H_1  \wh \Sigma_{xx}^* (\1_M\otimes \theta_0) = \1_M\otimes \theta_0$ and $ \wh B \wh H_1  \wh \Sigma_{xx}^* (\1_M\otimes \theta_0) = \wh \Sigma_{xx}^* (\1_M\otimes \theta_0)$.
Next, we can rewrite $ \wh B $ as $ \wh B = J^{M-1}\otimes \wh \Sigma_{xx} + \wh R $, where $ J \in \mR^{M\times M}$ is the permutation matrix defined in \eqref{permutation}, and
\begin{align*}
	\wh R = \begin{bmatrix}
		\0&\0&\cdots&\0&\0&\widehat{\Sigma}_{xx}^{(1)} - \wh\Sigma_{xx}\\
		\widehat{\Sigma}_{xx}^{(2)} -\wh\Sigma_{xx}&\0&\cdots&\0&\0&\0\\
		\0&\widehat{\Sigma}_{xx}^{(3)} -\wh \Sigma_{xx}&\cdots&\0&\0&\0\\
		\vdots&\vdots&\vdots&\vdots&\vdots&\vdots\\
		\0&\0&\cdots&\0&\widehat{\Sigma}_{xx}^{(M)} - \wh\Sigma_{xx}&\0\\
	\end{bmatrix}\in \mR^{q\times q}.
\end{align*}
By central limit theorem, it can be shown that $ \|\wh R \| = O_p(n^{-1/2}) $.
As derived in Appendix A.6, we have $ H_2 = Q \otimes I_p $, where $ Q = (a, J^{M-1}a,\dots, Ja)^\top \in \mR^{M\times M} $ for some vector $ a\in \mR^M $ satisfying $ \1_M^\top a = \0 $.
Then we can verify that $ \1_M\1_M^\top J^{M-1} Q =  \1_M\1_M^\top Q = \0$.
Consequently, it follows that $ \wh H_1 \wh B H_2 \wh B \wh H_1 \wh B H_2 =   \wh H_1 \wh R H_2 \wh R \wh H_1 \wh R H_2$.
Thus, we should have $ \|\wh H_1 \wh B H_2 \wh B \wh H_1 \wh B H_2\| = O_p(n^{-3/2}) $.
Finally, we obtain that $ \wh \mu(\alpha) = \1_M\otimes \theta_0 + O_p(\alpha^2 n^{-3/2}) $.

\textsc{Step 2.} By law of large numbers, we have $ \Sigma_{xx} \to_p \Sigma_{xx} $ and $ \Sigma_{xx}^{(m)} \to_p \Sigma_{xx} $ as $ N\to \infty $ for each $ 1\le m\le M $. Then we can obtain $ \wh B \to_p B$ and $\wh H_1 \to_p H_1 $ as $ N\to \infty $, where $ B = J^{M-1} \otimes \Sigma_{xx} $ and $ H_1 = P_1 \Sigma_{xx}^{-1} P_1^\top $. By replacing $ \wh B $ and $ \wh H_1 $ by $ B $ and $ H_1 $ in \eqref{Omega_inv} respectively, we can derive that $ \Omega^{-1}  = \alpha^{-1} H_1 + H_2 - \alpha H_2 B H_2 +O(\alpha^2) $. Here, we have used the fact that $ H_1 B H_2 = H_2 B H_1 = \0$. Consequently, we have $ V(\alpha)  = H_1 (I_M\otimes \Sigma_{xx}) H_1^\top + \alpha^2 H_2 (I_M\otimes \Sigma_{xx}) H_2^\top +O(\alpha^4) = M^{-1} (\1_M\1_M^\top) \otimes \Sigma_{xx}^{-1} + \alpha^2 (QQ^\top) \otimes \Sigma_{xx} +O(\alpha^4)$.

Let $ e_m \in \mR^M $ be the unit vector with $ 1 $ in the $ m $-th entry. Then we should have $ \wh \theta^{(m)} = (e_m \otimes \1_p)^\top \wh \theta^* $. In addition, $  (e_m \otimes \1_p)^\top \wh \mu(\alpha) = \theta_0 + O_p(\alpha^2 n^{-3/2}) $, and $ (e_m \otimes \1_p)^\top V(\alpha)(e_m \otimes \1_p) = M^{-1}\Sigma_{xx}^{-1} +\alpha^2 \|a\|^2 \Sigma_{xx} +O(\alpha^4) $, where $ \|a\| $ can be calculated by \eqref{eq:a_gamma>0} and \eqref{eq:a_gamma=0} in Appendix A.6. Thus, $  \sqrt{n} \big\{(\wh \theta^{(m)} - \theta_0 + O_p(\alpha^2 n^{-3/2}) \big\} \to_d N_p\Big(\0, \sigma^2 \big\{M^{-1}\Sigma_{xx}^{-1} +\alpha^2 \|a\|^2 \Sigma_{xx} +O(\alpha^4) \big\} \Big) $. Consequently,
\begin{equation*}
	 \sqrt{N}  \Big\{(\wh \theta^{(m)} - \theta_0 + O_p(\alpha^2 n^{-3/2}) \Big\} \to_d N_p\bigg(\0, \sigma^2 \Big\{\Sigma_{xx}^{-1} +\alpha^2 M \|a\|^2 \Sigma_{xx} +O(\alpha^4) \Big\} \bigg).
\end{equation*}
Further note that $ \|a\| $ is equal to the $ l_2 $-norm of the first row of $ H_2 = Q\otimes I_p $. We thus complete the proof.

\scsubsection{Appendix A.6 The Explicit form of $ Q $}

One can verify that $ A = \wt A \otimes I_p$, where
\begin{equation*}
	\wt A = \begin{bmatrix}
		1 & 0 &0&  \cdots &0& \gamma & -(1+\gamma)\\
	-(1+\gamma) & 1  &0&\cdots &0& 0 & \gamma\\
	\gamma & -(1+\gamma)  &1&\cdots & 0 & 0&0\\
	\vdots & \vdots  & \vdots & \vdots & \vdots& \vdots& \vdots\\
	0 & 0 &0&\cdots  &\gamma& -(1+\gamma) & 1
	\end{bmatrix}\in\mR^{M\times M}
\end{equation*}
is a circulant matrix.
Then we should have $ H_2 = P_2 (A_{22}^P)^{-1} P_2^\top = (Z\otimes I_p) \big\{ (Z^\top \otimes I_p)(\wt A\otimes I_p)(Z\otimes I_p) \big\}^{-1}  (Z^\top \otimes I_p) =Q\otimes I_p$, where $ Q = Z(Z^\top \wt A Z)^{-1}Z^\top $, and $ Z $ is defined in \eqref{Z_mat}.
It can be easily verified that $Q\wt{A}Q = Q$, which implies $Q$ is the generalized inverse of $\wt{A}$. Then it suffices to find the generalized inverse of $\wt{A}$. Since $\rank(\wt{A}) = M-1$ and $\1_M$ is the eigenvector of zero eigenvalue, according to Theorem 2 in \cite{searle1979inverting}, $Q$ has the following form
\begin{equation}\label{1}
	Q = \{\wt{A} + (1+\gamma)\1_M\1_M^\top\}^{-1} - \frac{1}{(1+\gamma)M^2}\1_M\1_M^\top.
\end{equation}
Since the inverse of a circulant matrix is also a circulant matrix, we must have $\{\wt{A} + (1+\gamma)\1_M\1^\top_M\}^{-1}$ be a circulant matrix, which has the form of
\begin{equation*}
\{\wt{A} + (1+\gamma)\1_M\1^\top_M\}^{-1} =
\begin{bmatrix}
	x_0 & x_1 & x_2 & \cdots & x_{M-1}\\
	x_{M-1} & x_0 & x_1 & \cdots & x_{M-2}\\
	x_{M-2} & x_{M-1} & x_0 & \cdots & x_{M-3}\\
	\vdots & \vdots & \vdots & \ddots & \vdots\\
	x_1 & x_2 & x_3 & \cdots & x_0
\end{bmatrix}.
\end{equation*}
Since $\{\wt{A} + (1+\gamma)\1_M\1^\top_M\}\{\wt{A} + (1+\gamma)\1_M\1^\top_M\}^{-1} = I_M$, we use the first row of $\{\wt{A} + (1+\gamma)\1_M\1^\top_M\}$ to multiply all columns of $\{\wt{A} + (1+\gamma)\1_M\1^\top_M\}^{-1}$. Then we have the following $M$ equations,
\begin{equation*}
\begin{cases}
	(2+\gamma)x_0 + (1+\gamma)x_{M-1} + \cdots + (1+\gamma)x_3+ (1+2\gamma)x_2 + 0x_1 = 1\\
	(2+\gamma)x_1 + (1+\gamma)x_0 + \cdots + (1+\gamma)x_4+ (1+2\gamma)x_3 + 0x_2 = 0\\
	\cdots\\
	(2+\gamma)x_{M-1} + (1+\gamma)x_{M-2} + \cdots + (1+\gamma)x_2+ (1+2\gamma)x_1 + 0x_0 = 0.
\end{cases}
\end{equation*}
Sum up all the equations above, we have $(1+\gamma)\sum^{M-1}_{m=0}x_m = 1/M$. Use all $M$ equations above minus $(1+\gamma)\sum^{M-1}_{m=0}x_m = 1/M$, we then have
\begin{equation}\label{2}
\begin{cases}
	x_0 - (1+\gamma)x_1 + \gamma x_2  =(x_0 - x_1) - \gamma(x_1 - x_2) = 1-\frac{1}{M}\\
	x_1 - (1+\gamma)x_2 + \gamma x_3  =(x_1 - x_2) - \gamma(x_2 - x_3) = -\frac{1}{M}\\
	\cdots\\
	x_{M-1} - (1+\gamma)x_0 + \gamma x_1  =(x_{M-1} - x_0) - \gamma(x_0 - x_1) = -\frac{1}{M}.
\end{cases}
\end{equation}
We then study the cases of $ \gamma>0 $ and $ \gamma = 0 $ subsequently.

{\sc Case 1 $( \gamma>0) $.} From equations at \eqref{2}, we can verify that
\begin{equation}\label{3}
x_k - x_{k+1} = \frac{1}{\gamma^k}(x_{0}-x_{1})  + \frac{1}{M}S_k - \frac{1}{\gamma^k} \text{\ for\ } k = 1,2,...,M-1,
\end{equation}
where $S_m = \sum^m_{k=1}\gamma^{-k}$.
Sum up all the equations in \eqref{3}, we have
\begin{equation*}
0 = \sum^{M-1}_{m=0}(x_m - x_{m+1}) =\gamma S_M (x_0-x_1) + \frac{1}{M}\sum^{M-1}_{m=1}S_m - S_{M-1}.
\end{equation*}
Therefore, we can obtain that
\begin{equation} \label{eq:beta}
\beta = x_0-x_1 = \frac{MS_{M-1} - \sum^{M-1}_{m=1}S_m}{M\gamma S_M}.
\end{equation}
Again, by \eqref{3} we can verify that
\begin{equation}\label{4}
x_m = x_0 - \beta\gamma S_m - \frac{1}{M}\sum^{m-1}_{k=1}S_k + S_{m-1} \text{\ for\ } m=2,...,M-1.
\end{equation}
Thus we can verify that
\begin{equation*}
\frac{1}{(1+\gamma)M} = \sum^{M-1}_{m=0}x_m = Mx_0 -a\gamma\sum^{M-1}_{m=1}S_m -\frac{1}{M}\sum^{M-1}_{m=2}\sum^{m-1}_{k=1}S_k + \sum^{M-2}_{m=1}S_m,
\end{equation*}
which implies
\begin{equation*}
x_0 = \frac{1}{(1+\gamma)M^2} +\frac{\beta\gamma}{M}\sum^{M-1}_{m=1}S_m +\frac{1}{M^2}\sum^{M-1}_{m=2}\sum^{m-1}_{k=1}S_k - \frac{1}{M}\sum^{M-2}_{m=1}S_m.
\end{equation*}
Since $Q$ is also a circulant matrix, we assume its first row is $(a_0,...,a_{M-1})$. According to \eqref{1} and \eqref{4}, we have
\begin{equation}\label{eq:a_gamma>0}
	\begin{cases}
		a_0 =& x_0 - \frac{1}{(1+\gamma)M^2} = \frac{\beta\gamma}{M}\sum^{M-1}_{m=1}S_m +\frac{1}{M^2}\sum^{M-1}_{m=2}\sum^{m-1}_{k=1}S_k - \frac{1}{M}\sum^{M-2}_{m=1}S_m\\
		a_1 =& a_0 - \beta\\
		a_m =& a_0 -\beta\gamma S_m - \frac{1}{M}\sum^{m-1}_{k=1}S_k + S_{m-1} \text{\ for\ } m=2,...,M-1,
	\end{cases}
\end{equation}
where $ \beta $ is defined in \eqref{eq:beta}.
Then we can write $ Q = (a, J^{M-1}a,\dots, Ja)^\top $, where $ a=(a_0,...,a_{M-1})^\top \in \mR^M $, and $ J \in \mR^{M\times M}$ is a permutation matrix defined as
\begin{equation} \label{permutation}
	J = \begin{bmatrix}
		0&1&0&\cdots&0&0\\
		0&0&1&\cdots&0&0\\
		\vdots&\vdots&\vdots&\vdots&\vdots&\vdots\\
		0&0&0&\cdots&0&1\\
		1&0&0&\cdots&0&0
	\end{bmatrix}.
\end{equation}

{\sc Case 2 $ (\gamma  = 0) $.} From equations at \eqref{2}, we have
\begin{equation*}
	x_0 - x_1 = 1-\frac{1}{M},\ x_1 - x_2 = -\frac{1}{M},\ \dots, x_{M-1} - x_0 = -\frac{1}{M}.
\end{equation*}
By similar procedures used above, we can obtain that $ x_0 = (M-1)/(2M) + 1/M^2 $ and $ x_m = x_0 - (M-m)/M $ for $ m=1,\dots, M-1 $.
Then by \eqref{1}, we have
\begin{equation}\label{eq:a_gamma=0}
	\begin{cases}
		a_0 =& x_0 - \frac{1}{M^2}= \frac{M-1}{2M} \\
		a_m =& a_0 - \frac{M-m}{M} = \frac{2m - M - 1}{2M} \text{ for } m=1,\dots, M-1 .
	\end{cases}
\end{equation}

Additionally, by \eqref{1} and the fact that $(1+\gamma) \sum_{m=0}^{M-1}x_m = 1/M $, it can be easily verified that $ \1_M^\top a = \sum_{m=0}^{M-1}a_m = 0$ for any $ \gamma\ge 0 $.

\scsection{APPENDIX B}

\begin{lemma} \label{lemma:invertible}
Suppose $ \Sigma_{xx} \in \mR^{p \times p} $ is a symmetric positive definite matrix with eigenvalues $ \lambda_1\ge\cdots\ge\lambda_p>0 $. Consider a block matrix defined by
\begin{align*}
	\Omega = \begin{bmatrix}
		I_p&\0&\0&\cdots&\0&\gamma I_p&-\Delta\\
		-\Delta&I_p&\0&\cdots&\0&\0&\gamma I_p\\
		\gamma I_p&-\Delta&I_p&\0&\cdots&\0&\0\\
		\vdots&\vdots&\vdots&\vdots&\vdots &\vdots&\vdots\\
		\0&\0&\0&\cdots&\gamma I_p&-\Delta&I_p
	\end{bmatrix}\in \mR^{q\times q},
\end{align*}
where $ \Delta = (1+\gamma) I_p - \alpha \Sigma_{xx} \in \mR^{p \times p}$. Assume $ \gamma \ne 1$, $ \alpha\ne 0 $ and $ \alpha \ne 2(1+\gamma)/\lambda_j $ for each $ 1\le j \le p $. Then we have $ \Omega $ is invertible.
\end{lemma}
\begin{proof}
To prove the invertibility of $ \Omega $, we only need to prove that $ 0 $ is not an eigenvalue of $ \Omega $ under the assumed conditions. To this end, we first note that $ \Omega $ is a block circulant matrix. Hence, we can analyze the eigenvalues of $ \Omega $ with help of the results in \cite{kaveh2011block}. Specifically, we denote $ \omega_m =  (e^{2\pi i /M})^{m-1} \in \mathbb{C}$ for $ 1\le m\le M $, where $ i =\sqrt{-1}$ is the imaginary unit. We can see that $ |\omega_m|=1 $ and $ \omega_m^M=1 $ for each $ 1\le m \le M $. Define $ H(\omega) = I_p + \omega^{M-2} \gamma I_p - \omega^{M-1} \Delta$ as a map from $ \mathbb{C} $ to $ \mathbb{C}^{p \times p} $. Then by the results in \cite{kaveh2011block}, the eigenvalues of $ \Omega $ will be the union of the eigenvalues of matrices $ H(\omega_m)$ for $1\le m\le M $. We next show that $ 0 $ cannot be an eigenvalue of $ H(\omega_m) $ for any $ 1\le m\le M $. Since $  \Sigma_{xx} $ has eigenvalues $ \lambda_1, \dots, \lambda_p $, the eigenvalues of $ H(\omega_m) $ should be $ h(\omega_m) =1 + \omega_m^{M-2} \gamma - \omega_m^{M-1} (1+\gamma - \alpha \lambda_j)$ for $ 1\le j\le p $.  Suppose $ h(\omega_m) = 0 $ for some $ m $. Then we have $0=\omega_m^2 h(\omega_m) = \omega_m^2 + (\alpha\lambda_j - 1- \gamma) \omega_m + \gamma$, i.e., $ \omega_m $ is the root of the quadratic equation $ \omega^2 + (\alpha\lambda_j - 1- \gamma) \omega + \gamma=0$ with real coefficients. (1) Suppose $ \omega_m \notin\mR$. Then we must have $|\omega_m|^2= \omega_m \widebar{\omega}_m  = \gamma \ne 1$. This is a contradiction. (2) Suppose $ \omega_m  \in \mR$ . Then it must be $ -1 $ or $ 1 $. But this is impossible under the conditions $ \alpha\ne 0 $ and $ \alpha \ne 2(1+\gamma)/\lambda_j $ for each $ 1\le j \le p $. Thus, we have proved the result.	
\end{proof}

\begin{lemma}\label{lemma:SR}
Suppose $ \Sigma_{xx} \in \mR^{p \times p} $ is a symmetric positive definite matrix with eigenvalues $ \lambda_1\ge\cdots\ge\lambda_p>0 $. Consider a block matrix defined $ D = [\Delta, -\gamma I_p; I_p, \0] \in \mR^{(2p)\times (2p)}$, where $ \Delta = (1+\gamma)I_p - \alpha \Sigma_{xx}$, $ \alpha>0 $ and $ \gamma\ge0 $.
\begin{enumerate} [(a)]
	\item Suppose that $ 0<\alpha<2(1+\gamma)/\lambda_{1}  $. Then, $ \gamma\le\rho(D)<1 $ if $ 0\le \gamma <1 $, $ \rho(D)>1 $ if $ \gamma>1 $, and $ \rho(D)=1 $ if $ \gamma=1 $.

	\item If $ \gamma=0 $, then $\rho(D) =   \max\{|1-\alpha\lambda_1|,  |1-\alpha\lambda_p| \} $.
	When $ \alpha = 2/ (\lambda_{1}+\lambda_p) $, $ \rho(D) $ attains the minimum value as $\rho(D)=(\lambda_1 - \lambda_p) /  (\lambda_1 + \lambda_p)$.
	
	\item If $ \gamma \in [(1-\sqrt{\alpha\lambda_j })^2, (1+\sqrt{\alpha\lambda_j })^2] $ for all $ 1\le j\le p $, then $ \rho(D) =  \sqrt{\gamma} $.
	When $ \alpha = 4/(\sqrt{\lambda_{1} }+ \sqrt{\lambda_p})^2 $ and $ \gamma = (\sqrt{\lambda_{1}} - \sqrt{\lambda_p})^2 / (\sqrt{\lambda_{1}} + \sqrt{\lambda_p})^2 $, then $ \rho(D) = (\sqrt{\lambda_{1}} - \sqrt{\lambda_p}) / (\sqrt{\lambda_{1}} + \sqrt{\lambda_p})$.

	\item Suppose $ 0<\alpha< 1/\lambda_1$. If $ \gamma = 0 $, then $ \rho(D) = 1-\alpha \lambda_p $; If $0<  \gamma <(1-{\alpha \lambda_p})^2 $, then $\rho < 1-\alpha \lambda_p$. When $ \gamma = (1-\sqrt{\alpha \lambda_p})^2$, $ \rho(D)  $ attains the minimum value as $\rho(D)= 1-\sqrt{\alpha \lambda_p} $.
	
	\item If $ \alpha> 2(1+\gamma)/ \lambda_1 $ or $ \gamma>1 $, then $ \rho(D)>1 $.

\end{enumerate}
\end{lemma}
\begin{proof} We prove (a)-(e) subsequently.

\noindent\textbf{Proof of (a).} We only consider the case of $0\le \gamma <1 $, since the case of $ \gamma\ge1 $ can be derived similarly. Suppose $ \xi \in \mathbb{C}$ is an arbitrary eigenvalue of $ D $, and $ v = (v_1^\top, v_2^\top)^\top \in \mathbb{C}^{2p} $ is the corresponding eigenvector, where $v_1\in \mathbb{C}^p  $ and $v_2\in \mathbb{C}^p  $. By $ Dv = \xi v $, we have $ \Delta v_1 - \gamma v_2 = \xi v_1 $ and $ v_1 = \xi v_2 $.
(1) If $ \gamma = 0 $, one can derive that the nonzero eigenvalues of $ D $ are the same as those of $ \Delta $. Then by assumed condition $ 0<\alpha<2(1+\gamma)/\lambda_{1}  $, we should have $ \rho(D) = \rho(\Delta)=\max_{1\le j\le p}|1-\alpha\lambda_j|<1  $ for all $ 1\le j\le p $. Consequently, $ \rho(D) \in [0, 1) = [\gamma, 1) $.
(2)  If $ 0<\gamma<1  $, one can verify that $ D $ is nonsingular. Thus $ \xi \ne 0 $. Then we have $ \Delta v_1 = (\gamma / \xi + \xi)v_1 $. In other words, $ \gamma / \xi + \xi $ is the eigenvalue of the real symmetric matrix $ \Delta $. Hence  $ \gamma / \xi + \xi  $ must be a real number. (2.1) Suppose that $ \xi \notin \mR$, i.e., $ \xi = a+bi $ with $ b \ne 0 $. Since $ \gamma / \xi + \xi  = (\gamma / |\xi|^2) \widebar{\xi} + \xi \in \mR $, we should have $ |\xi|^2 = \gamma $. Thus $ |\xi| = \sqrt{\gamma}$. (2.2) Suppose $ \xi \in \mR $. By Weyl's inequality and condition $ 0<\alpha < 2(1+\gamma)/\lambda_1$ , we know that the eigenvalues of $ \Delta$ are contained in the interval $ \big(-(1+\gamma), 1+\gamma \big) $. Thus we should have $-(1+\gamma) < (\gamma / \xi + \xi)<1+\gamma $. (2.2.1) If $ \xi>0 $, by $ (\gamma / \xi + \xi)<1+\gamma $ we have $ \xi^2 - (1+\gamma)\xi + \gamma = (\xi-1)(\xi-\gamma)<0 $, or equivalently, $ \xi \in (\gamma, 1) $. (2.2.2) If $ \xi<0 $, by $ -(1+\gamma) < (\gamma / \xi + \xi) $ we have $ \xi^2 + (1+\gamma)\xi + \gamma = (\xi+1)(\xi+\gamma)<0 $, or equivalently, $ \xi \in (-1, -\gamma) $. Thus, if $ \xi \in \mR $, then $ |\xi| \in (\gamma,1) $. Since $ \sqrt{\gamma}\in (\gamma, 1) $ if $ 0<\gamma<1 $, we conclude that $ \rho(D)\in (\gamma, 1) $.
By (1) and (2) above, we have proved that $ \rho(D) \in [\gamma, 1) $ if $ 0<\alpha<2(1+\gamma)/\lambda_{1}  $ and $ 0\le \gamma <1$.

\noindent\textbf{Proof of (b).}  Similar to the proof of (a), $ D $ has the same nonzero eigenvalues as those of $ \Delta $ if $ \gamma = 0 $. Since the eigenvalues of $ \Delta $ are $ 1 - \alpha \lambda_j$ for $ 1\le j \le p $, claims in (b) follow immediately.

\noindent \textbf{Proof of (c).} The proof can be found in \cite{polyak1964}. We give a simple version for the sake of completeness. If $ \gamma=0 $, then we must have $ \alpha \lambda_j = 1 $ for all $ 1\le j\le p $ under the assumed conditions. One can verify that $ \Delta $ will be a zero matrix in this case. Consequently, the eigenvalues of $ D $ are all zero, i.e., $ \rho(D)=0 $. We then consider the case of $ \gamma>0 $. First, note that when $ \gamma>0 $, $ D $ is nonsingular, i.e., $ D $ has no zero eigenvalue. Then we construct a matrix function of $ \lambda $ as
\begin{equation*}
	G(\lambda) = \begin{bmatrix}
		1+\gamma - \alpha \lambda& -\gamma\\
		1&0
	\end{bmatrix} \in \mR^{2\times 2}.
\end{equation*}
One can verify that the roots of characteristic polynomial $ \det(\xi I_{2p} - D) $ are the same as the roots of the $ p $ characteristic polynomials $ \det(\xi I_{2} - G(\lambda_j)) $ ($ 1\le j\le p $). Hence the eigenvalues of $ D $ are the union of the eigenvalues of $ G(\lambda_j) $ for $ 1\le j \le p $. Therefore, we turn to investigate the spectrum of $ G(\lambda) $. The characteristic polynomial of $ D $ is $ \det(\xi I_{2} - G(\lambda)) = \xi^2 -(1+\gamma - \alpha \lambda)\xi + \gamma$. We find that the discriminant of the corresponding quadratic equation $ (1+\gamma - \alpha \lambda)^2 - 4\gamma \le 0 $ if and only if $ \gamma \in  [(1-\sqrt{\alpha\lambda})^2, (1+\sqrt{\alpha\lambda})^2]$. In this case, the roots of this equation have the same magnitude $ \sqrt{\gamma} $. The first claim in (b) follows immediately. The second claim can be obtained by direct calculations for the given $ \alpha $ and $ \gamma $. Thus we have proved (c).

\noindent\textbf{Proof of (d).} The case of $ \gamma =0$ follows immediately from (b). We then investigate the case of $0< \gamma<(1-\alpha \lambda_p)^2 $.
(1) First, if $ \gamma \in \big[(1-\sqrt{\alpha \lambda_p})^2, (1-\alpha\lambda_p)^2 \big) $, then $ \gamma \in \big[  (1-\sqrt{\alpha \lambda_j})^2, (1+\sqrt{\alpha \lambda_j})^2 \big]$ for all $ 1 \le j \le p $ under the condition $ 0<\alpha< 1/\lambda_1$.
Then by (c), we know that $ \rho(D) = \sqrt{\gamma} $. Hence, $ \rho(D)<1-\alpha\lambda_p $ if $ \gamma \in \big[(1-\sqrt{\alpha \lambda_p})^2, (1-\alpha\lambda_p)^2 \big) $.
(2) We then consider the case of $0< \gamma <  (1-\sqrt{\alpha \lambda_p})^2 $.
In this case,  there must be some $ 1 \le j_1 \le p $ such that $ \gamma \ge  (1-\sqrt{\alpha \lambda_j})^2$ for $ j < j_1 $, and $ \gamma < (1-\sqrt{\alpha \lambda_j})^2$ for $ j \ge j_1 $.
(2.1) For $ j < j_1 $, it follows the proof of (c) that $ \rho(G(\lambda_j)) = \sqrt{\gamma} < 1-\sqrt{\alpha \lambda} <1- \alpha \lambda_p  $.
(2.2) For $ j\ge j_1 $, the characteristic equation  $  \xi^2 -(1+\gamma - \alpha \lambda_j)\xi + \gamma = 0 $ should have two different real roots $ \xi_1, \xi_2 $, since the discriminant of this equation is larger than $ 0 $.
Then by Vieta's formulas, we have $ \xi_1 +\xi_2 = 1 +\gamma - \alpha \lambda_j >0$ and $ \xi_1 \xi_2 = \gamma>0$. Thus, both $ \xi_1 $ and $ \xi_2 $ should be positive.
By (a),  we cam conclude that $ \xi_1, \xi_2 $ should be contained in $ [\gamma,1) $.
Since $ \xi_1 \xi_2 = \gamma  $, we can further verify that $ \xi_1, \xi_2 $ must be strictly larger that $ \gamma $, i.e., $ \xi_1, \xi_2 \in (\gamma,1) $.
Consequently, $ \xi_1 = (1+\gamma - \alpha \lambda_j) - \xi_2 < 1-  \alpha \lambda_j \le1- \alpha \lambda_p$. Obviously, we also have $  \xi_2 < 1-  \alpha \lambda_p$.
Therefore, we should have $ \rho(G(\lambda_j))<1-\alpha\lambda_p $ for all $ j_1< j \le p $.
It follows from (2.1) and (2.2) that $ \rho(D)<1-\alpha\lambda_p $ if $ \gamma \in \big(0, (1-\sqrt{\alpha \lambda_p})^2 \big) $.
Hence by (1) and (2), we have proved that $ \rho(D)<1-\alpha\lambda_p $ if $0< \gamma < (1-\alpha \lambda_p)^2 $.

From (1) above, we know that $ \rho(D) = 1-\sqrt{\alpha\lambda_p} $ if $ \gamma =  (1-\sqrt{\alpha\lambda_p})^2$, and $ \rho(D)>1-\sqrt{\alpha\lambda_p}  $ if $ \gamma \in \big((1-\sqrt{\alpha \lambda_p})^2, (1-\alpha\lambda_p)^2 \big) $.
Hence, to prove that the minimal $ \rho(D) $ is achieved by choosing $ \gamma =  (1-\sqrt{\alpha\lambda_p})^2$, we only need to verify that if $ 0< \gamma < (1-\sqrt{\alpha \lambda_p})^2 $, then $ \rho(D) > 1-\sqrt{\alpha\lambda_p} $.
In fact, if $0< \gamma < (1-\sqrt{\alpha \lambda_p})^2 $, one can verify that the equation  $h(\xi) =   \xi^2 -(1+\gamma - \alpha \lambda_p)\xi + \gamma = 0 $ has two different positive real roots $ \xi_1,\xi_2 $ by similar arguments in (2.2) above.
To prove that one of $ \xi_1 $ and $ \xi_2 $ is larger that $  1-\sqrt{\alpha \lambda_p}$, it suffices to show $h( 1-\sqrt{\alpha \lambda_p})<0$.
This is true, since it can be easily verified that $ h( 1-\sqrt{\alpha \lambda_p}) = - (1-\sqrt{\alpha \lambda_p})^2\sqrt{\alpha \lambda_{p}} + \gamma \sqrt{\alpha\lambda_p }<0$. Thus, we complete the proof of (d).

\noindent\textbf{Proof of (e).} By arguments used above, one can easily show that, if $ \gamma>1 $, then $ \rho(D)>1 $.
Hence, we only prove the result under $ \alpha> 2(1+\gamma) / \lambda_1 $ and $ 0\le\gamma\le 1 $.
The case of $ \gamma = 0 $ is also trivial, hence we only deal with the case of $ 0<\gamma\le 1 $.
Consider the characteristic equation $ \det(\xi I_{2} - G(\lambda_1)) = \xi^2 + b\xi + \gamma=0$, where $ b =-( 1+\gamma-\alpha \lambda_1 ) $. If $ \alpha > 2(1+\gamma) / \lambda_1 $, we have $b > 1+\gamma  $. Hence, the discriminant of this equation $ b^2 - 4\gamma > (1-\gamma)^2\ge 0 $. Since $ 0<\gamma\le 1 $,  one root of this equation should be $  (-b - \sqrt{b^2 - 4\gamma}) /2 < \{-(1+\gamma) - (1-\gamma) \}/ 2= -1$, i.e., $ \rho(G(\lambda_{1}))>1 $. Thus, in either case, we should have $ \rho(D)>1 $. We complete the proof.
\end{proof}

\newpage


%
%

\renewcommand{\bibsection}{}
\scsection{REFERENCES}
\bibliographystyle{apalike}
\bibliography{ref}

\end{document}